\documentclass[conference]{IEEEtran}

\usepackage[cmex10]{amsmath}
\usepackage{amssymb}
\usepackage{algorithm}
\usepackage{algorithmic}
\usepackage{url}

\newtheorem{lem}{Lemma}
\newtheorem{prop}{Proposition}
\newtheorem{claim}{Claim}

\allowdisplaybreaks
\begin{document}

\title{Heavy Traffic Optimal Resource Allocation Algorithms for Cloud Computing Clusters}

\author{\IEEEauthorblockN{Siva Theja Maguluri and R. Srikant}
\IEEEauthorblockA{Department of ECE and CSL\\
University of Illinois at Urbana-Champaign\\
siva.theja@gmail.com; rsrikant@illinois.edu}
\and
\IEEEauthorblockN{Lei Ying}
\IEEEauthorblockA{Department of ECEE\\
Arizona State University\\
lying6@asu.edu}
}

\maketitle
\begin{abstract}
Cloud computing is emerging as an important platform for business, personal and mobile computing applications. In this paper, we study a stochastic model of cloud computing, where jobs arrive according to a stochastic process and request resources like CPU, memory and storage space. We consider a model where the resource allocation problem can be separated into a routing or load balancing problem and a scheduling problem. We study the join-the-shortest-queue routing and power-of-two-choices routing algorithms with MaxWeight scheduling algorithm. It was known that these algorithms are throughput optimal. In this paper, we show that these algorithms are queue length optimal in the heavy traffic limit.
\end{abstract}


\section {Introduction}\label{sec:intro}
Cloud computing services are emerging as an important resource for personal as well as commercial computing applications. Several cloud computing systems are now commercially available, including Amazon EC2 system \cite{EC2}, Google's AppEngine \cite{appengine}, and Microsoft's Azure \cite{azure}. A comprehensive survey on cloud computing can be found in \cite{FosZhaRai_08,ArmFoxGri_09,MenNgo_09}.

In this paper, we focus on cloud computing platforms that provide infrastructure as service. Users submit requests for resources in the form of virtual machines (VMs). Each request specifies the amount of resources it needs in terms of processor power, memory, storage space, etc.. We call these requests jobs. The cloud service provider first queues these requests and then schedules them on physical machines called servers.

Each server has a limited amount of resources of each kind. This limits the number and types of jobs that can be scheduled on a server. The set of jobs of each type that can be scheduled simultaneously at a server is called a configuration. The convex hull of the possible configurations at a server is the capacity region of the server. The total capacity region of the cloud is then the Minkowski sum of the capacity regions of all servers.

The simplest architecture for serving the jobs is to queue them at a central location. In each time slot, a central scheduler chooses the configuration at each server and allocates jobs to the servers, in a preemptive manner. As pointed out in \cite{magsriyin12}, this problem is then identical to scheduling in an ad hoc wireless network with interference constraints.
In practice, however, jobs are routed to servers upon arrival. Thus, queues are maintained at each individual server. It was shown in \cite{magsriyin12} that join-the-shortest queue-type algorithms for routing, along with the MaxWeight scheduling algorithm \cite{TasEph_92} at each server is throughput optimal. The focus of this paper is to study the delay, or equivalently, the queue length performance of the algorithms presented in \cite{magsriyin12}.


Characterizing the exact delay or queue length in general is difficult. So, we study the system in the heavy-traffic regime, i.e., when the exogenous arrival rate is close to the boundary of the capacity region. In this regime, for some systems, the multi-dimensional state of the system reduces to a single dimension, called state-space collapse. In \cite{Bramson_state_space, Williams_state_space}, a method was outlined to use the state-space collapse for studying the diffusion limits of several queuing systems. This procedure has been successfully applied to a variety of multiqueue models served by multiple servers \cite{rei_state_space, har_state_space, harlop_state_space, belwil_state_space}. But these models assume that the system is work conserving, i.e., queued jobs are processed at maximum rate by each server. Stolyar \cite{Sto_04}, generalized this notion of state-space collapse and resource pooling to a generalized switch model, where it is hard to define work-conserving policies. This was used to establish the heavy traffic optimality of the MaxWeight algorithm.

Most of these results are based on considering a scaled version of queue lengths and time, which converges to a regulated Brownian motion, and then show sample-path optimality in the scaled time over a finite time interval. This then allows a natural conjecture about steady state distribution.
 In \cite{erysri-heavytraffic}, the authors present an alternate method to prove heavy traffic optimality  that is not only simpler, but shows heavy traffic optimality in unscaled time. In addition, this method directly obtains heavy-traffic optimality in steady state. The method consists of the following three steps.
\begin{enumerate}
\item[(1)] \emph{Lower bound:} First a lower bound is obtained on the weighted sum of expected queue lengths by comparing with a single-server queue. A lower bound for the single-server queue, similar to the Kingman bound \cite{kingman}, then gives a lower bound to the original system.
\item[(2)] \emph{State-space collapse}: The second step is to show that the state of the system collapses to a single dimension. Here, it is not a complete state-space collapse, as in the Brownian limit approach, but an approximate one. In particular, this step is to show that the queue length along a certain direction  increases as the exogenous arrival rate gets closer to the boundary of the capacity region but the queue length in any perpendicular direction is bounded.
\item[(3)] \emph{Upper bound}: The state-space collapse is then used to obtain an upper bound on the weighted queue length. This is obtained by using a natural Lyapunov function suggested by the resource pooling. Heavy-traffic optimality can be obtained if  the lower bounds and the upper bounds coincide.

\end{enumerate}

In this paper, we apply the above three-step procedure to study the resource allocation algorithms presented in \cite{magsriyin12}. We briefly review the results in \cite{magsriyin12} now. Jobs are first routed to the servers, and are then queued at the servers, and a scheduler schedules jobs at each server. So, we need an algorithm that has two components, viz.,
\begin{enumerate}
\item a \emph{routing algorithm} that routes  new jobs to servers in each time slot (we assume that the
jobs are assigned to a server upon arrival and they cannot be moved
to a different server) and
\item a \emph{scheduling algorithm} that chooses the configuration of each server, i.e., in each time
slot, it decides which jobs to serve. Here
we assume that jobs can be preempted, i.e., a job can be served in
a time slot, and then be preempted if it is not scheduled in the
next time slot. Its service can be resumed in the next time it is
scheduled. Such a model is applicable in situations where job sizes are typically large.
\end{enumerate}



It was shown in \cite{magsriyin12} that using the join-the-shortest-queue (JSQ) routing and MaxWeight scheduling algorithm is throughput optimal. In Section \ref{sec:JSQMaxwt}, we show that this policy is queue length optimal in the heavy traffic limit  when all the servers are identical. We use the three step procedure described above to prove the heavy traffic optimality. The lower bound in this case is identical to the case of the MaxWeight scheduling problem. However, state-space collapse does not directly follow from the corresponding results for the MaxWeight algorithm in \cite{erysri-heavytraffic} due to the additional routing step here.  We use this to obtain an upper bound that coincides with the lower bound in the heavy traffic limit.

JSQ needs queue length information of all servers at the router. In practice, this communication
overhead can be quite significant when the number of servers is large. An alternative algorithm is the power-of-two-choices routing algorithm. In each time slot, two servers are chosen uniformly at random and new arrivals are routed to the server with the shorter queue. It was shown in \cite{magsriyin12} that the power-of-two-choices routing algorithm with the MaxWeight scheduling is throughput optimal if all the servers are identical.  Here, we show that the heavy-traffic optimality in this case is a minor modification of the corresponding result for JSQ routing and MaxWeight scheduling.

A special case of the resource allocation problem is when all the jobs are of same type. In this case, scheduling is not required at each server. The problem reduces to a routing-only problem which is well studied \cite{Mit_96,Yilupoweroftwo, Chen10power2, He_down_po2, dobrushin_po2}. For reasons to be explained later, the results, from Section \ref{sec:JSQMaxwt} cannot be applied in this case since the capacity region is along a single dimension (of the form $\lambda<\mu$). In Section \ref{sec:powerof2}, we show heavy traffic optimality of the power-of-two-choices routing algorithm. The lower and upper bounds in this case are identical to the case of JSQ routing in \cite{erysri-heavytraffic}. The main contribution here is to show state-space collapse, which is somewhat different compared to \cite{erysri-heavytraffic}. The results here complement the heavy-traffic optimality results in \cite{Chen10power2, He_down_po2} which were obtained using Brownian motion limits.

\subsection*{Note on Notation}
The set of real numbers, the set of non-negative real numbers,

and the set of positive real numbers are denoted by $\mathbb{R}$, $\mathbb{R}_+$ and $\mathbb{R}_{++}$ respectively.
We denote vectors in $\mathbb{R}^{J}$ or $\mathbb{R}^{M}$ by $x$, in normal font. We
use bold font $\mathbf{x}$ to denote vectors in $\mathbb{R}^{JM}$.
Dot product in the vector spaces $\mathbb{R}^{J}$ or $\mathbb{R}^{M}$ is denoted by $\left\langle x,y\right\rangle$ and
the dot product in $\mathbb{R}^{JM}$ is denoted by $\left\langle \mathbf{x},\mathbf{y}\right\rangle $.

\section{System Model and Algorithm}\label{sec:model}

Consider a discrete time cloud computing system as follows. There are $M$
servers indexed by $m$. Each server has $I$ different kinds of resources
such as processing power, disk space, memory, etc.. Server $m$ has $R_{i,m}$
units of resource $i$ for $i\in\{1,2,3,...,I\}$. There are $J$ different types
of jobs indexed by $j$. Jobs of type $j$ need $r_{i,j}$ units of
resource $i$ for their service. A job is said to be of size $D$ if it takes $D$ units of time to finish its service. Let $D_{max}$ be the maximum allowed service time.

Let ${\cal A}_j(t)$ denote the set of type-$j$ jobs that arrive at the beginning of time slot $t.$ Indexing the jobs in ${\cal A}_j(t)$ from $1$ through $|{\cal A}_j(t)|,$
we define $a_j(t)=\sum_{k \in{\cal A}_j(t)} D_k,$ to be the overall size of the jobs in ${\cal A}_j(t)$ or the total time slots requested by the jobs in ${\cal A}_j(t)$. 
Thus, $a_j(t)$ denotes the total work load of type $j$ that arrives in time slot $t$.
We assume that $a_j(t)$ is a stochastic process which is i.i.d. across time slots, $\mathbb{E}[a_j(t)]=\lambda_j$ and $\Pr(a_j(t)=0)>\epsilon_A$ for some $\epsilon_A>0$ for all $j$ and $t$. Many of these assumptions can be relaxed, but we make these assumptions for the ease of exposition. Second moments of the arrival processes
are assumed to be bounded. Let $var[a_{j}(t)]=\sigma_{j}^2$, $\lambda=(\lambda_{1},....\lambda_{J})$
and $\sigma=(\sigma_{1},....\sigma_{J})$. We denote $\sigma^{2}=(\sigma_{1}^{2},....\sigma_{J}^{2})$.

In each time slot, the central router routes the new arrivals to one of the servers. Each server maintains $J$ queues corresponding to the work loads of the  $J$ different types of jobs.
Let $q_{j,m}(t)$ denote the total backlogged job size of the type $j$ jobs at  server $m$ at time slot $t$.

Consider server $m$. We say that server $m$ is in configuration $s=(s_{1},s_{2},...,s_{J})\in\left(\mathbb{Z}_{+}\right)^{J}$
if the server is serving $s_{1}$ jobs of type $1$, $s_{2}$ jobs
of type $2$ etc. This is possible only if the server has enough
resources to accommodate all these jobs. In other words, $\overset{J}{\underset{j=1}{\sum}}s_{j}r_{i,j}\leq R_{i,m}\forall i\in\{1,2,...,I\}$. Let $s_{max}$ be the maximum number of jobs of any type that can be scheduled on any server.
Let $\mathcal{S}_{m}$ be the set of feasible configurations on server
$m$. We say that $s$ is a maximal configuration if no other job
can be accommodated i.e., for every $j'$ $s+e_{j'}$ (where $e_{j'}$
is the unit vector along $j'$) violates at least one of the resource
constraints. Let $\mathcal{C}_{m}^{*}$ be the convex hull of the maximal
configurations of server $m$. Let $\mathcal{C}_{m}=\{s\in\left(\mathbb{R}_{+}\right)^{J}:s\leq s^{*} \textnormal{ for some } s^{*}\in\mathcal{C}_{m}^{*}\}$.
Here $s\leq s^{*}$ means $s_{j}\leq s_{j}^{*}\forall j\in\{1,2,...,J\}$.
$\mathcal{C}_{m}$ can be thought of as the capacity region for server
$m$. Note that if $\lambda\in interior(\mathcal{C}_{m})$, there exists an $\epsilon>0$ such that $\lambda(1+\epsilon)\in\mathcal{C}_{m}$.
$\mathcal{C}_{m}$ is a convex polytope in the nonnegative quadrant
of $\mathbb{R}^{J}$.

Define $\mathcal{C}= \underset{m=1}{\overset{M}{\sum}}\mathcal{C}_{m} = \{s\in\left(\mathbb{R}_{+}\right)^{J}:\exists s^{m}\in\mathcal{C}_{m} \ \forall \ m \textrm{ s.t. }s\leq\overset{M}{\underset{m=1}{\sum}}s^{m}\}$.
We denote this as $\mathcal{C}=\underset{m=1}{\overset{M}{\sum}}\mathcal{C}_{m}$.
Here $s^m$ just denotes an element in $\mathcal{C}_{m}$ and not $m^\textnormal{th}$ power of $s$.
Then, $\mathcal{C}=\underset{m=1}{\overset{M}{\sum}}\mathcal{C}_{m}$,
where $\sum$ denotes the Minkowski sum of sets. So, $\mathcal{C}$
 is again a convex polytope in the nonnegative quadrant of $\mathbb{R}^{J}$.
So, $\mathcal{C}$ can be described by a set of hyperplanes as follows:\[
\mathcal{C}=\{s\geq0:\left\langle c^{(k)},s\right\rangle \leq b^{(k)},k=1,...K\}\]
where $K$ is the number of hyperplanes that completely defines $\mathcal{C}$,
and $(c^{(k)},b^{(k)})$ completely defines the $k^{th}$ hyperplane
$\mathcal{H}^{(k)}$, $\left\langle c^{(k)},s\right\rangle = b^{(k)}$. Since $\mathcal{C}$ is in the first quadrant, we
have \[
||c^{(k)}||=1\quad,c^{(k)}\geq0,\quad b^{(k)}\geq0\quad for\: k=1,2,...K.\]

It was shown in \cite{magsriyin12} that $\mathcal{C}$ is the capacity region of this system. Similar to $\mathcal{C}$, define $\mathcal{S}=\underset{m=1}{\overset{M}{\sum}}\mathcal{S}_{m}$.

\begin{lem}
\label{lem:boundary}\textup{Given the $k^{th}$ hyperplane $\mathcal{H}^{(k)}$
of the capacity region $\mathcal{C}$ (i.e., $\left\langle c^{(k)},\lambda\right\rangle =b^{(k)}$),
for each server $m$, there is a $b_{m}^{(k)}$ }such that \textup{$\left\langle c^{(k)},\lambda\right\rangle =b_{m}^{(k)}$
is the boundary of the capacity region $\mathcal{C}_{m}$, and $b^{(k)}=\overset{M}{\underset{m=1}{\sum}}b_{m}^{(k)}$}
. Moreover, for \textup{every set $\left\{ \lambda_{m}^{(k)}\in\mathcal{C}_{m}\right\} _{m}$
such that $\lambda^{(k)}=\overset{M}{\underset{m=1}{\sum}}\lambda_{m}^{(k)}$
and $\lambda^{(k)}\in\mathcal{C}$ lies on the $k^{th}$ hyperplane
$\mathcal{H}^{(k)}$ , we have that $\left\langle c^{(k)},\lambda_{m}^{(k)}\right\rangle =b_{m}^{(k)}$.}\end{lem}
\begin{proof}
Define $b_{m}^{(k)}=\underset{s\in\mathcal{C}_{m}}{\max}\left\langle c^{(k)},s\right\rangle $.
Then, since \\${\mathcal{C}=\underset{m=1}{\overset{M}{\sum}}\mathcal{C}_{m}}$,
we have that $b^{(k)}=\overset{M}{\underset{m=1}{\sum}}b_{m}^{(k)}$.

Again, by the definition of $\mathcal{C}$, for every $\lambda\in\mathcal{C}$,
there are $\lambda_{m}^{(k)}\in\mathcal{C}_{m}$ for each $m$ such
that $\lambda^{(k)}=\overset{M}{\underset{m=1}{\sum}}\lambda_{m}^{(k)}$.
However, these may not be unique. We will prove that for every such
$\left\{ \lambda_{m}^{(k)}\right\} _{m}$, for each $m$, $\left\langle c^{(k)},\lambda_{m}^{(k)}\right\rangle =b_{m}^{(k)}$.
Suppose, for some server $m_{1}$, $\left\langle c^{(k)},\lambda_{m_{1}}^{(k)}\right\rangle <b_{m_{1}}^{(k)}$.
Then since $\left\langle c^{(k)},\underset{m=1}{\overset{M}{\sum}}\lambda_{m}^{(k)}\right\rangle =\underset{m=1}{\overset{M}{\sum}}b_{m}^{(k)}$,
there exists $m_{2}$ such that \\${\left\langle c^{(k)},\lambda_{m_{2}}^{(k)}\right\rangle >b_{m_{2}}^{(k)}}$
which is a contradiction. Thus, we have the lemma.
\end{proof}

\section{JSQ Routing and MaxWeight \\Scheduling}\label{sec:JSQMaxwt}

In this section, we will study the performance of JSQ routing with MaxWeight scheduling, as described in Algorithm \ref{alg:JSQMaxwt}.

\begin{algorithm}

\caption{\label{alg:JSQMaxwt}JSQ Routing and MaxWeight Scheduling}

\begin{enumerate}
\item \emph{Routing Algorithm:} All the type $j$ arrivals in a time slot
are routed to the server with the smallest queue  length for type $j$ jobs,
i.e., the server $m_{j}^{*}=\underset{m\in\{1,2,...M\}}{\arg\min}q_{j,m}$. Ties are broken uniformly at random.
\item \emph{Scheduling Algorithm:} In each time slot, server $m$ chooses
a configuration $\overline{s^{m}}\in\mathcal{C}_{m}^{*}$ so that $\overline{s^{m}}=\underset{\overline{s^{m}}\in\mathcal{C}_{m}^{*}}{\arg\max}\underset{j=1}{\overset{J}{\sum}}\overline{s_{j}^{m}}q_{j,m}$.
It then schedules up to a maximum of $\overline{s_{j}^{m}}$ jobs of type $j$ (in a preemptive manner). Note that even if the queue length is greater than the allocated service, all of it may not be utilized, e.g., when the backlogged size is from a single job, since different chunks of the same job cannot be scheduled simultaneously. Denote the actual number of jobs chosen by $s_{j}^{m}$. Note that if $q_{j,m}\geq D_{max}s_{max}$, then  $\overline{s_{j}^{m}}= s_{j}^{m}$.
%
\end{enumerate}

\end{algorithm}

Let ${Y}_{j,m}(t)$ denote the state of the queue for type-$j$
jobs at server $m$, where ${Y}_{j,m}^{i}(t)$ is the (backlogged) size of the $i^{{\rm th}}$
type-$j$ job at server $m$. It is easy to see that $\mathbf{Y}(t)=\{{Y}_{j,m}(t)\}_{j,m}$
is a Markov chain under the JSQ routing and MaxWeight scheduling. Then, $q_{j,m}(t)=\sum_{i}{Y}_{j,m}^{i}(t)$ is a function of the state ${Y}_{j,m}(t)$.

The queue lengths of workload evolve according to the following equation:
\begin{align}
q_{j,m}(t+1) & =q_{j,m}(t)+a_{j,m}(t)-s_{j}^{m}(t)\nonumber \\
 & =q_{j,m}(t)+a_{j,m}(t)-\overline{s_{j}^{m}}(t)+\overline{u}_{j,m}(t)\label{eq:q}\end{align}
where $\overline{u}_{j,m}(t)$ is the unused service, given by $\overline{u}_{j,m}(t) = \overline{s_{j}^{m}}(t)-s_{j}^{m}(t)$,  $\overline{s_{j}^{m}}(t)$ is the MaxWeight schedule and $s_{j}^{m}(t)$ is the actual schedule chosen by the scheduling algorithm and the arrivals are\begin{equation}
a_{j,m}(t)=\begin{cases}
a_{j}(t) & \textrm{ if }m=m_{j}^{*}(t)\\
0 & \textrm{otherwise}\end{cases}.\label{eq:arrival}\end{equation}
Here, $m_{j}^{*}$ is the server
chosen by the routing algorithm for type $j$ jobs. Note that \begin{equation}
\overline{u}_{j,m}(t)=0\textrm{ when }q_{j,m}(t)+a_{j,m}(t)\geq D_{max}s_{max}.\label{eq:unused}\end{equation}
Also, denote $s=(s_{j})_{j}$ where
\begin{equation} s_{j}=\underset{m=1}{\overset{M}{\sum}}s_{j}^{m}. \label{eq:s_j}
\end{equation}
Denote $\mathbf{a}=(a_{j,m})_{j,m}$, $\mathbf{s}=(s_{j}^{m})_{j,m}$
and $\mathbf{\overline{u}}=(\overline{u}_{j,m})_{j,m}$. Also denote $\mathbf{1}$ to be the vector with $1$ in all components.

 It was shown in \cite{magsriyin12} that this algorithm is throughput optimal. Here, we will show that this algorithm is heavy traffic optimal.

Recall that the capacity region is bounded by $K$ hyperplanes, each
hyperplane $\mathcal{H}^{(k)}$ described by its normal vector $c^{(k)}$
and the value $b^{(k)}$. Then, for any $\lambda\in interior(\mathcal{C})$,
we can define the distance of $\lambda$ to $\mathcal{H}^{(k)}$ and
the closest point, respectively, as\begin{align}
\epsilon^{(k)} & =\underset{s\in\mathcal{H}^{(k)}}{\min}||\lambda-s||\label{eq:epsilon}\\
\lambda^{(k)} & =\lambda+\epsilon^{(k)}c^{(k)}\nonumber\end{align}
where $\epsilon^{(k)}>0$ for each $k$ since $\lambda\in interior(\mathcal{C})$.
We let $\epsilon\triangleq\left(\epsilon^{(k)}\right)_{k=1}^{K}$
denote the vector of distances to all hyperplanes. Note that $\lambda^{(k)}$
may be outside the capacity region $\mathcal{C}$ for some hyperplanes.
So define\[
\mathcal{K}_{\lambda}\triangleq\left\{ k\in\{1,2,...K\}:\lambda^{(k)}\in\mathcal{C}\right\} \]
 $\mathcal{K}_{\lambda}$ identifies the set of \emph{dominant hyperplanes}
whose closest point to $\lambda$ is on the boundary of the capacity
region $\mathcal{C}$ hence is a feasible average rate for service.
Note that for any $\lambda\in interior(\mathcal{C})$, the set $\mathcal{K}_{\lambda}$
is non-empty, and hence is well-defined. We further define \[
\mathcal{K}_{\lambda}^{o}\triangleq\left\{ k\in\mathcal{K}_{\lambda}^{}:\lambda^{(k)}\in Relint(\mathcal{F}^{(k)})\right\} \]
where $\mathcal{F}^{(k)}$ denotes the face on which $\lambda^{(k)}$
lies and $Relint$ means relative interior. Thus, $\mathcal{K}_{\lambda}^{o}$
is the subset of faces in $\mathcal{K}_{\lambda}$ for which the projection
of $\lambda$ is not shared by more than one hyperplane.

For $\epsilon\triangleq\left(\epsilon^{(k)}\right)_{k=1}^{K}>0$,
let $\lambda^{(\epsilon)}$ be the arrival rate in the interior of the capacity region so that its distance from the hyperplane $\mathcal{H}^{(k)}$
is $\epsilon^{(k)}$. Let $\lambda^{(k)}$ be the closest point to
$\lambda^{(\epsilon)}$ on $\mathcal{H}^{(k)}$. Thus, we have \begin{equation}
\lambda^{(k)}=\lambda^{(\epsilon)}+\epsilon^{(k)}c^{(k)}.\label{eq:lambda}\end{equation}
Let $\mathbf{q}^{(\epsilon)}(t)$ be the queue length
process when the arrival rate is $\lambda^{(\epsilon)}$.

Define $\mathbf{c}^{(k)}\in\mathbb{R}_{+}^{JM}$, indexed
by $j,m$ as $\mathbf{c}_{j,m}=\frac{c_{j}}{\sqrt{M}}$. We expect that the state space collapse occurs along the direction $\mathbf{c}^{(k)}$. This is intuitive. For a fixed $j$, JSQ routing tries to equalize the queue lengths across servers. For a fixed server $m$, we expect that the state space collapse occurs along $c^{(k)}$ when approaching the hyperplane $\mathcal{H}^{(k)}$, as shown in \cite{erysri-heavytraffic}. Thus, for JSQ routing and MaxWeight, we expect that the state space collapse occurs along $\mathbf{c}^{(k)}$ in $\mathbb{R}^{JM}$.

For each $k\in\mathcal{K}_{\lambda^{(\epsilon)}}^{o},$ define the projection
and perpendicular component of $\mathbf{q}^{(\epsilon)}$ to the
vector $\mathbf{c}^{(k)}$ as follows:\begin{align*}
\mathbf{q}_{||}^{(\epsilon,k)} & \triangleq\left\langle \mathbf{c}^{(k)},\mathbf{q}^{(\epsilon)}\right\rangle \mathbf{c}^{(k)}\\
\mathbf{q}_{\bot}^{(\epsilon,k)} & \triangleq\mathbf{q}^{(\epsilon)}-\mathbf{q}_{||}^{(\epsilon,k)}\end{align*}

In this section, we will prove the following proposition.

\begin{prop}\label{prop:main}
Consider the cloud computing system described in Section \ref{sec:model}. Assume all the servers are identical, i.e., $R_{i,m}=R_i$ for all servers $m$ and resources $i$ and that
JSQ routing and MaxWeight scheduling as described in Algorithm \ref{alg:JSQMaxwt} is used. Let
the exogenous arrival rate be $\lambda^{(\epsilon)}\in Interior(\mathcal{C})$
and the standard deviation of the arrival vector be $\sigma^{(\epsilon)}\in \mathbb{R}^J_{++}$ where the
parameter $\epsilon=\left(\epsilon^{(k)}\right)_{k=1}^{K}$ is so
that $\epsilon^{(k)}$ is the distance of $\lambda^{(\epsilon)}$
from the $k^{th}$ hyperplane $\mathcal{H}^{(k)}$ as defined in (\ref{eq:epsilon}). Then  for each $k\in\mathcal{K}_{\lambda^{(\epsilon)}}^{o}$, the steady state queue length satisfies
\begin{align*}
\epsilon^{(k)}\mathbb{E}\left[\left\langle \mathbf{c}^{(k)},\mathbf{q}(t)\right\rangle \right] & \leq  \frac{\zeta^{(\epsilon,k)}}{2}+B_{2}^{(\epsilon,k)}\end{align*}
where $\zeta^{(\epsilon,k)}=\frac{1}{\sqrt{M}}\left\langle \left(c^{(k)}\right)^{2},\left(\sigma^{(\epsilon)}\right)^{2}\right\rangle +\frac{\left(\epsilon^{(k)}\right)^{2}}{\sqrt{M}}$, $B_{2}^{(\epsilon,k)}$ is $o(\frac{1}{{\epsilon}^{(k)}})$

In the heavy traffic limit as $\epsilon^{(k)}\downarrow0$,
this bound is tight, i.e.,
\begin{equation*}
\underset{\epsilon^{(k)}\downarrow0}{\lim}\epsilon^{(k)}\mathbb{E}\left[\left\langle \mathbf{c}^{(k)},\mathbf{q}^{(\epsilon)}\right\rangle \right]=\frac{\zeta^{(k)}}{2}\end{equation*}
where  $\zeta^{(k)}=\frac{1}{\sqrt{M}}\left\langle \left(c^{(k)}\right)^{2},\left(\sigma\right)^{2}\right\rangle $.
\end{prop}

We will prove this proposition by following the three step procedure described in Section \ref{sec:intro}, by first obtaining a lower bound, then showing state space collapse and finally using the state space collapse result to obtain an upper bound.

\subsection{Lower Bound}\label{sub:lower_bound}

Since $\lambda^{(\epsilon)}$ is in the interior of $\mathcal{C}$,
the process $\left\{ \mathbf{q}^{(\epsilon)}(t)\right\} _{t}$
has a steady state distribution. We will obtain a lower bound on $\mathbb{E}\left[\left\langle \mathbf{c}^{(k)},\mathbf{q}^{(\epsilon)}\right\rangle \right]=\mathbb{E}\left[\underset{j=1}{\overset{J}{\sum}}\frac{c_{j}^{(k)}}{\sqrt{M}}\left(\overset{M}{\underset{m=1}{\sum}}q_{jm}\right)\right]$ in steady state as follows.

Consider the single server queuing system, $\phi^{(\epsilon)}(t)$
with arrival process $\frac{1}{\sqrt{M}}\left\langle c^{(k)},a^{(\epsilon)}(t)\right\rangle $ and
service process given by $\frac{b^{(k)}}{\sqrt{M}}$ at each time
slot. Then $\phi(t)$ is stochastically smaller than $\left\langle \mathbf{c}^{(k)},\mathbf{q}(t)^{(\epsilon)}\right\rangle $.
Thus, we have \[\mathbb{E}\left[\left\langle \mathbf{c}^{(k)},\mathbf{q}^{(\epsilon)}\right\rangle \right]\geq\mathbb{E}\left[\phi^{(\epsilon)}\right].\]
Using  $\phi^{2}$ as Lyapunov function for the single server queue and noting that the drift of it should be zero in steady state, one can bound $\mathbb{E}\left[\overline{\phi}^{(\epsilon)}\right]$ as follows \cite{erysri-heavytraffic}
\[
\epsilon^{(k)}\mathbb{E}\left[\overline{\phi}^{(\epsilon)} \right]\geq\frac{\zeta^{(\epsilon,k)}}{2}-B_{1}^{(\epsilon,k)}.\]
where $\left(c^{(k)}\right)^{2}=\left(\left(c_{j}^{(k)}\right)^{2}\right)_{j=1}^{J}$, $B_{1}^{(\epsilon,k)}=\frac{b^{(k)}\epsilon^{(k)}}{2}$
and
$\zeta^{(\epsilon,k)}=\frac{1}{\sqrt{M}}\left\langle \left(c^{(k)}\right)^{2},\left(\sigma^{(\epsilon)}\right)^{2}\right\rangle +\frac{\left(\epsilon^{(k)}\right)^{2}}{\sqrt{M}}$.

Thus, in the heavy traffic limit as $\epsilon^{(k)}\downarrow0$,
we have that \begin{equation}
\underset{\epsilon^{(k)}\downarrow0}{\lim}\epsilon^{(k)}\mathbb{E}\left[\left\langle \mathbf{c}^{(k)},\overline{\mathbf{q}}^{(\epsilon)}\right\rangle \right]\geq\frac{\zeta^{(k)}}{2}\label{eq:lower_bound}\end{equation}
where $\zeta^{(k)}=\frac{1}{\sqrt{M}}\left\langle \left(c^{(k)}\right)^{2},\left(\sigma\right)^{2}\right\rangle $.

\subsection{State Space Collapse}\label{sub:jsqmw_state_space}

In this subsection, we will show that there is a state space collapse along the direction $\mathbf{c}^{(k)}$. We know that as the arrival rate approaches the boundary of the capacity region, i.e., $\epsilon^{(k)}\rightarrow0$, the steady state mean queue length $\mathbb{E}[||\mathbf{q}||]\rightarrow \infty$. We will show that as $\epsilon^{(k)}\rightarrow0$, queue length projected along any direction perpendicular to $\mathbf{c}^{(k)}$ is bounded. So the constant does not contribute to the  first order term in $\frac{1}{\epsilon^{(k)}}$, in which we are interested. Therefore, it is sufficient to study a bound on the queue length along  $\mathbf{c}^{(k)}$. This is called state-space collapse.

Define the following Lyapunov functions.
\begin{align*}
V(\mathbf{q})&\triangleq\overset{M}{\underset{m=1}{\sum}}\underset{j=1}{\overset{J}{\sum}}q_{j,m}^{2},\;
  W_{\bot}^{(k)}(\mathbf{q})\triangleq\left\Vert \mathbf{q}_{\bot}^{(k)}\right\Vert, \;
   W_{||}^{(k)}(\mathbf{q})\triangleq\left\Vert \mathbf{q}_{||}^{(k)}\right\Vert \\
V_{||}^{(k)}(\mathbf{q})&\triangleq\left\langle \mathbf{c}^{(k)},\mathbf{q}^{(\epsilon)}\right\rangle ^{2}=\left\Vert \mathbf{q}_{||}^{(k)}\right\Vert ^{2}=\frac{1}{M}\left(\overset{M}{\underset{m=1}{\sum}}\underset{j=1}{\overset{J}{\sum}}q_{j,m}c_{j}\right)^{2}.
\end{align*}
Define the drift of the above Lyapunov functions.
 \begin{align*}
\Delta V(\mathbf{q}) & \triangleq\left[V(\mathbf{q}(t+1))-V(\mathbf{q}(t))\right]\mathcal{I}(\mathbf{q}(t)=\mathbf{q})\\
\Delta W_{\bot}^{(k)}(\mathbf{q}) & \triangleq\left[W_{\bot}^{(k)}(\mathbf{q}(t+1))-W_{\bot}^{(k)}(\mathbf{q}(t))\right]\mathcal{I}(\mathbf{q}(t)=\mathbf{q})\\
\Delta W_{||}^{(k)}(\mathbf{q}) & \triangleq\left[W_{||}^{(k)}(\mathbf{q}(t+1))-W_{||}^{(k)}(\mathbf{q}(t))\right]\mathcal{I}(\mathbf{q}(t)=\mathbf{q})\\
\Delta V_{||}^{(k)}(\mathbf{q}) & \triangleq\left[V_{||}^{(k)}(\mathbf{q}(t+1))-V_{||}^{(k)}(\mathbf{q}(t))\right]\mathcal{I}(\mathbf{q}(t)=\mathbf{q})
\end{align*}

To show the state space collapse happens along the direction
of $\mathbf{c}^{(k)}$, we will need a result by Hajek \cite{hajek_drift}, which gives a bound on $\left\Vert \mathbf{q}_{\bot}^{(k)}\right\Vert $ if the drift of $W_{\bot}^{(k)}(\mathbf{q})$ is negative. Here we use the following special case of the result by Hajek, as presented in \cite{erysri-heavytraffic}.

\begin{lem} \label{lem:Hajek}
For an irreducible and aperiodic Markov Chain $\{X[t]\}_{t\geq 0}$ over a countable state space $\mathcal{X},$ suppose $Z:\mathcal{X}\rightarrow \mathbb{R}_+$ is a nonnegative-valued Lyapunov function. We define the drift of $Z$ at $X$ as $$\Delta Z(X) \triangleq [Z(X[t+1])-Z(X[t])]\>\mathcal{I}(X[t]=X),$$ where $\mathcal{I}(.)$ is the indicator function. Thus, $\Delta Z(X)$ is a random variable that measures the amount of change in the value of $Z$ in one step, starting from state $X.$ This drift is assumed to satisfy the following conditions:
\begin{enumerate}
\item There exists an $\eta>0,$ and a $\kappa<\infty$ such that for all $X\in \mathcal{X}$ with $ Z(X)\geq \kappa,$
 \begin{eqnarray*}
    \mathbb{E}[\Delta Z(X) | X[t]=X]
    \leq - \eta .
 \end{eqnarray*}
\item There exists a $D < \infty$ such that for all $X \in \mathcal{X},$
  \begin{eqnarray*}
  \mathbb{P}\left(|\Delta Z(X)| \leq D\right) = 1.
  \end{eqnarray*}
\end{enumerate}
Then, there exists a $\theta^\star > 0$ and a $C^\star < \infty$ such that
$$\limsup_{t\rightarrow \infty} \mathbb{E}\left[e^{\theta^\star Z(X[t])}\right] \leq C^\star.$$
If we further assume that the Markov Chain $\{X[t]\}_t$ is positive recurrent, then $Z(X[t])$ converges in distribution to a random variable $\bar{Z}$ for which
$$ \mathbb{E}\left[e^{\theta^\star \bar{Z}}\right] \leq C^\star,$$
which directly implies that all moments of $\bar{Z}$ exist and are finite.
\end{lem}

We also need Lemma 7 from \cite{erysri-heavytraffic}, which gives the drift of $W_{\bot}^{(k)}(\mathbf{q})$ in terms of drifts of $V(\mathbf{q})$ and $V_{||}^{(k)}(\mathbf{q})$.
\begin{lem}
Drift of $W_{\bot}^{(k)}$ can be bounded as follows:\begin{equation}
\Delta W_{\bot}^{(k)}(\mathbf{q})\leq\frac{1}{2\left\Vert \mathbf{q}_{\bot}^{(k)}\right\Vert }(\Delta V(\mathbf{q})-\Delta V_{||}^{(k)}(\mathbf{q}))\quad \forall \   \mathbf{q}\in\mathbb{R}_{+}^{J}\label{eq:W_perp}\end{equation}

\end{lem}
Let us first consider the last term in this inequality.

\begin{align}
\lefteqn{\mathbb{E}\left.\left[\vartriangle V_{||}^{(k)}(\mathbf{q}^{(\epsilon)})\right|\mathbf{q}^{(\epsilon)}(t)=\mathbf{q}^{(\epsilon)}\right]}\nonumber \\
= & \mathbb{E}\left[\left.V_{||}^{(k)}(\mathbf{q}^{(\epsilon)}(t+1))-V_{||}^{(k)}(\mathbf{q}^{(\epsilon)}(t))\right|\mathbf{q}^{(\epsilon)}(t)=\mathbf{q}^{(\epsilon)}\right]\nonumber \\
= & \mathbb{E}\left[\left.\left\langle \mathbf{c}^{(k)},\mathbf{q}^{(\epsilon)}(t+1)\right\rangle ^{2}
-\left\langle \mathbf{c}^{(k)},\mathbf{q}^{(\epsilon)}(t)\right\rangle ^{2}\right|\mathbf{q}(t)=\mathbf{q}^{(\epsilon)}\right]\nonumber \\
= & \mathbb{E}\left[\left\langle \mathbf{c}^{(k)},\mathbf{q}^{(\epsilon)}(t)+\mathbf{a}^{(\epsilon)}(t)-\mathbf{s}^{(\epsilon)}(t)+\overline{\mathbf{u}}^{(\epsilon)}(t)\right\rangle ^{2}\right.\nonumber\\
&\left.\left.-\left\langle \mathbf{c}^{(k)},\mathbf{q}^{(\epsilon)}(t)\right\rangle ^{2}\right|\mathbf{q}(t)=\mathbf{q}^{(\epsilon)}\right]\nonumber \\
= & \mathbb{E}\left[\left\langle \mathbf{c}^{(k)},\mathbf{q}^{(\epsilon)}(t)+\mathbf{a}^{(\epsilon)}(t)-\mathbf{s}^{(\epsilon)}(t)\right\rangle ^{2}+\left\langle \mathbf{c}^{(k)},\overline{\mathbf{u}}^{(\epsilon)}(t)\right\rangle ^{2}\right.\nonumber \\
 & +2\left\langle \mathbf{c}^{(k)},\mathbf{q}^{(\epsilon)}(t)+\mathbf{a}^{(\epsilon)}(t)-\mathbf{s}^{(\epsilon)}(t)\right\rangle \left\langle \mathbf{c}^{(k)},\overline{\mathbf{u}}^{(\epsilon)}(t)\right\rangle\nonumber\\
  & \left. \left.-\left\langle \mathbf{c}^{(k)},\mathbf{q}^{(\epsilon)}(t)\right\rangle ^{2}\right|\mathbf{q}(t)=\mathbf{q}^{(\epsilon)}\right]\nonumber \\
\geq & \mathbb{E}\left[\left\langle \mathbf{c}^{(k)},\mathbf{a}^{(\epsilon)}(t)-\mathbf{s}^{(\epsilon)}(t)\right\rangle ^{2}-2\left\langle \mathbf{c}^{(k)},\mathbf{s}^{(\epsilon)}(t)\right\rangle \left\langle \mathbf{c}^{(k)},\overline{\mathbf{u}}^{(\epsilon)}(t)\right\rangle\right.\nonumber \\
&\left.\left.+2\left\langle \mathbf{c}^{(k)},\mathbf{q}^{(\epsilon)}(t)\right\rangle \left\langle \mathbf{c}^{(k)},\mathbf{a}^{(\epsilon)}(t)-\mathbf{s}^{(\epsilon)}(t)\right\rangle \right|\mathbf{q}(t)=\mathbf{q}^{(\epsilon)}\right] \nonumber \\
\geq & 2\left\langle \mathbf{c}^{(k)},\mathbf{q}^{(\epsilon)}\right\rangle \left(\left\langle \mathbf{c}^{(k)},\mathbb{E}\left[\left.\mathbf{a}^{(\epsilon)}(t)\right|\mathbf{q}(t)=\mathbf{q}^{(\epsilon)}\right]\right.\right.\nonumber\\
&\left.\left.-\mathbb{E}\left[\left.\mathbf{s}^{(\epsilon)}(t)\right|\mathbf{q}(t)=\mathbf{q}^{(\epsilon)}\right]\right\rangle \right)-2\left\langle \mathbf{c}^{(k)},s_{max}\mathbf{1}\right\rangle ^2 \nonumber \\
= & \frac{2||\mathbf{q}_{||}^{(\epsilon,k)}||}{\sqrt{M}}\underset{j=1}{\overset{J}{\sum}}c_{j}\left(\overset{M}{\underset{m=1}{\sum}}\mathbb{E}\left[a_{j,m}^{(\epsilon)}(t)|\mathbf{q}(t)=\mathbf{q}^{(\epsilon)}\right]\right.\nonumber\\
&\left.-\overset{M}{\underset{m=1}{\sum}}\mathbb{E}\left[s_{j}^{m(\epsilon)}(t)|\mathbf{q}(t)=\mathbf{q}^{(\epsilon)}\right]\right)-K_{2}\nonumber \\
= & \frac{2||\mathbf{q}_{||}^{(\epsilon,k)}||}{\sqrt{M}}\underset{j=1}{\overset{J}{\sum}}c_{j}\left(\lambda_{j}^{(\epsilon)}-\overset{M}{\underset{m=1}{\sum}}\mathbb{E}\left[s_{j}^{m(\epsilon)}(t)|\mathbf{q}(t)=\mathbf{q}^{(\epsilon)}\right]\right)-K_{2}\label{eq:arrival-1}\\
= & \frac{2||\mathbf{q}_{||}^{(\epsilon,k)}||}{\sqrt{M}}\underset{j=1}{\overset{J}{\sum}}c_{j}\left(\lambda_{j}^{(k)}-\epsilon^{(k)}c_{j}^{(k)}\right.\nonumber\\
&\left.-\overset{M}{\underset{m=1}{\sum}}\mathbb{E}\left[s_{j}^{m(\epsilon)}(t)|\mathbf{q}(t)=\mathbf{q}^{(\epsilon)}\right]\right)-K_{2}\label{eq:lambda_epsilon}\\
= & \frac{2||\mathbf{q}_{||}^{(\epsilon,k)}||}{\sqrt{M}}\underset{j=1}{\overset{J}{\sum}}c_{j}\left(\overset{M}{\underset{m=1}{\sum}}\lambda_{j}^{m(k)}-\overset{M}{\underset{m=1}{\sum}}\mathbb{E}\left[s_{j}^{m(\epsilon)}(t)|\mathbf{q}(t)=\mathbf{q}^{(\epsilon)}\right]\right)\nonumber\\
&-K_{2}-\frac{2\epsilon^{(k)}}{\sqrt{M}}||\mathbf{q}_{||}^{(\epsilon,k)}||\label{eq:big_lambda}\\
= & \frac{2||\mathbf{q}_{||}^{(\epsilon,k)}||}{\sqrt{M}}\overset{M}{\underset{m=1}{\sum}}\underset{j=1}{\overset{J}{\sum}}c_{j}\left(\lambda_{j}^{m(k)}-\mathbb{E}\left[s_{j}^{m(\epsilon)}(t)|\mathbf{q}(t)=\mathbf{q}^{(\epsilon)}\right]\right)\nonumber\\
&-K_{2}-\frac{2\epsilon^{(k)}}{\sqrt{M}}||\mathbf{q}_{||}^{(\epsilon,k)}||\nonumber\\
\geq & -K_{2}-\frac{2\epsilon^{(k)}}{\sqrt{M}}||\mathbf{q}_{||}^{(\epsilon,k)}||\label{eq:V_parallel}
\end{align}
where $K_{2}=2JMs_{max}^{2}$. Equation (\ref{eq:arrival-1}) follows
from the fact that the sum of arrival rates at each server is same
as the external arrival rate. Equation (\ref{eq:lambda_epsilon}) follows
from (\ref{eq:lambda}). From the definition of $\mathcal{C}$, we
have that there exists $\lambda^{m(k)}\in\mathcal{C}_{m}$ such that
$\lambda^{(k)}=\overset{M}{\underset{m=1}{\sum}}\lambda^{m(k)}$.
This gives (\ref{eq:big_lambda}). From Lemma \ref{lem:boundary},
we have that for each $m$, there exists $b_{m}^{(k)}$ such that
$\underset{j=1}{\overset{J}{\sum}}c_{j}\lambda_{j}^{m(k)}=b_{m}^{(k)}$
and $\left\langle c^{(k)},s^{m(\epsilon)}\right\rangle \leq b_{m}^{(k)}$
for every $s^{m(\epsilon)}(t)\in\mathcal{C}_{m}$. Therefore, we have, for each $m$,
\[\underset{j=1}{\overset{J}{\sum}}c_{j}\left(\lambda_{j}^{m(k)}-\mathbb{E}\left[s_{j}^{m(\epsilon)}(t)|\mathbf{q}(t)=\mathbf{q}^{(\epsilon)}\right]\right)\geq0\]
and so (\ref{eq:V_parallel}) is true.

Now, let us consider the first term in (\ref{eq:W_perp}). By expanding the drift of $V(\mathbf{q}^{(\epsilon)})$ and using (\ref{eq:unused}), it can be easily seen that
\begin{align}
\lefteqn{\mathbb{E}\left[\vartriangle V(\mathbf{q}^{(\epsilon)})|\mathbf{q}^{(\epsilon)}(t)=\mathbf{q}^{(\epsilon)}\right]}\nonumber\\
\leq &
 K'+\mathbb{E}\left[\overset{M}{\underset{m=1}{\sum}}\underset{j=1}{\overset{J}{\sum}}\left(2q^{(\epsilon)}_{j,m}\left(a_{j,m}(t)-s_{j}^{m}(t)\right)\right)\right]\label{eq:bounded_moments}\end{align}
where $K'=M\left(\underset{j}{\sum}\left(\lambda_j^2+\sigma_j^2\right)+2Js_{max}^2(1+D_{max})\right)$

By definition of $a_{j,m}(t)$, (\ref{eq:arrival}) we have \begin{align}
\lefteqn{\mathbb{E}\left[\overset{M}{\underset{m=1}{\sum}}\underset{j=1}{\overset{J}{\sum}}2q^{(\epsilon)}_{j,m}a_{j,m}(t)\right] }\nonumber\\
 & =\mathbb{E}\left[\underset{j=1}{\overset{J}{\sum}}2q^{(\epsilon)}_{j,m_{j}^{*}}a_{j}(t)\right]\nonumber \\
 & =\underset{j=1}{\overset{J}{\sum}}2q^{(\epsilon)}_{j,m_{j}^{*}}\lambda^{(\epsilon)}_{j}\nonumber\\
 &\leq \underset{j=1}{\overset{J}{\sum}}2\lambda^{(\epsilon)}_{j}\overset{M}{\underset{m=1}{\sum}}\frac{q_{j,m}^{(\epsilon)}}{M} .\label{eq:arrival_rate}\end{align}
From (\ref{eq:bounded_moments})
and (\ref{eq:arrival_rate}), we have,
\begin{align}
\lefteqn{\mathbb{E}\left[\vartriangle V(\mathbf{q}^{(\epsilon)})|\mathbf{q}^{(\epsilon)}(t)=\mathbf{q}^{(\epsilon)}\right]}\nonumber\\
\leq & K'+\underset{j=1}{\overset{J}{\sum}}2\lambda^{(\epsilon)}_{j}\overset{M}{\underset{m=1}{\sum}}\frac{q_{j,m}^{(\epsilon)}}{M} -2\overset{M}{\underset{m=1}{\sum}}\mathbb{E}\left[\underset{j=1}{\overset{J}{\sum}}q_{j,m}^{(\epsilon)}s_{j}^{m}(t)\right]\label{eq:ref_po2_maxwt} \\
=& K'+\underset{j=1}{\overset{J}{\sum}}2\left(\lambda_{j}^{(k)}-\epsilon^{(k)}c_{j}^{(k)}\right)\overset{M}{\underset{m=1}{\sum}}\frac{q_{j,m}^{(\epsilon)}}{M} \nonumber\\
& -2\overset{M}{\underset{m=1}{\sum}}\mathbb{E}\left[\underset{j=1}{\overset{J}{\sum}}q_{j,m}^{(\epsilon)}s_{j}^{m}(t)\right] \nonumber\\
= & K'-\frac{2\epsilon^{(k)}}{\sqrt{M}}||\mathbf{q}_{||}^{(\epsilon,k)}||+2\overset{M}{\underset{m=1}{\sum}}\mathbb{E}\left[\underset{j=1}{\overset{J}{\sum}}q_{j,m}^{(\epsilon)}\left(\frac{\lambda_{j}^{(k)}}{M}-s_{j}^{m}(t)\right)\right]\nonumber \\
= & K_1 \hspace{-2pt} - \hspace{-2pt} \frac{2\epsilon^{(k)}}{\sqrt{M}}||\mathbf{q}_{||}^{(\epsilon,k)}||+2\overset{M}{\underset{m=1}{\sum}}\mathbb{E}\left[\underset{r^{m}\in\mathcal{C}_{m}}{\min}\underset{j=1}{\overset{J}{\sum}}q_{j,m}^{(\epsilon)}\left(\frac{\lambda_{j}^{(k)}}{M}-r_{j}^{m}\right)\right]\label{eq:maxwt}
\end{align}
where $K_1=K'+2JMD_{max}s_{max}^2$. Equation (\ref{eq:maxwt}) is true because of MaxWeight scheduling. Note that in algorithm \ref{alg:JSQMaxwt}, the actual service allocated to jobs of type $j$ at server $m$ is same as that of the MaxWeight schedule as long as the corresponding queue length is greater than $D_{max}s_{max}$. This gives the additional $2JMD_{max}s_{max}^2$ term.

 Assuming \emph{all the servers
are identical}, we have that for each $m$, $\mathcal{C}_{m}=\{\lambda/M:\lambda\in\mathcal{C}\}.$
So, $\mathcal{C}_{m}$ is a scaled version of $\mathcal{C}$. Thus, $\lambda^{m}=\lambda/M$. Since
$k\in\mathcal{K}_{\lambda^{(\epsilon)}}^{o},$ we also have that $k\in\mathcal{K}_{\lambda^{m(\epsilon)}}^{o}$
for the capacity region $\mathcal{C}_{m}$. Thus, there exists $\delta^{(k)}>0$
so that \[
\mathcal{B}_{\delta^{(k)}}^{(k)}\triangleq\mathcal{H}^{(k)}\cap\{r\in\mathbb{R}_{+}^{J}:||r-\lambda^{(k)}/M||\leq\delta^{(k)}\}\]
lies strictly within the face of $\mathcal{C}_{m}$ that corresponds
to $\mathcal{F}^{(k)}$. (Note that this is the only instance in the proof of Proposition \ref{prop:main} that we use the assumption that all the servers are identical.) Call this face $\mathcal{F}_{m}^{(k)}$. Thus we have,
\begin{align}
\lefteqn{\mathbb{E}\left[\vartriangle V(\mathbf{q}^{(\epsilon)})|\mathbf{q}^{(\epsilon)}(t)=\mathbf{q}^{(\epsilon)}\right]
-\left(K_1-\frac{2\epsilon^{(k)}}{\sqrt{M}}||\mathbf{q}_{||}^{(\epsilon,k)}||\right)}\nonumber\\
\leq & 2\overset{M}{\underset{m=1}{\sum}}\mathbb{E}\left[\underset{r^{m}\in\mathcal{B}_{\delta^{(k)}}^{(k)}}{\min}\underset{j=1}{\overset{J}{\sum}}q_{j,m}^{(\epsilon)}\left(\frac{\lambda_{j}^{(k)}}{M}-r_{j}^{m}\right)\right]\label{eq:ball}\\
= & 2\overset{M}{\underset{m=1}{\sum}}\mathbb{E}\left[\underset{r^{m}\in\mathcal{B}_{\delta^{(k)}}^{(k)}}{\min}\underset{j=1}{\overset{J}{\sum}}\left(q_{j,m}^{(\epsilon)}-\left\Vert \mathbf{q}_{||}^{(k)}\right\Vert \frac{c_{j}}{\sqrt{M}}\right)\left(\frac{\lambda_{j}^{(k)}}{M}-r_{j}^{m}\right)\right]\label{eq:get_q_perp}\\
= & 2\overset{M}{\underset{m=1}{\sum}}\mathbb{E}\left[\underset{r^{m}\in\mathcal{B}_{\delta^{(k)}}^{(k)}}{\min}\underset{j=1}{\overset{J}{\sum}}q_{\bot j,m}^{(\epsilon,k)}\left(\frac{\lambda_{j}^{(k)}}{M}-r_{j}^{m}\right)\right]\nonumber\\
= & -2\delta^{(k)}\overset{M}{\underset{m=1}{\sum}}\sqrt{\underset{j=1}{\overset{J}{\sum}}\left(q_{\bot j,m}^{(\epsilon,k)}\right)^{2}}\label{eq:delta}\\
\leq & -2\delta^{(k)}\sqrt{\overset{M}{\underset{m=1}{\sum}}\underset{j=1}{\overset{J}{\sum}}\left(q_{\bot j,m}^{(\epsilon,k)}\right)^{2}}\label{eq:norm}\\
= & -2\delta^{(k)}\left\Vert \mathbf{q}_{\bot}^{(k)}\right\Vert .\label{eq:drift_tight}
\end{align}
Equation (\ref{eq:get_q_perp}) is true because
 $c$ is a vector perpendicular to the face $\mathcal{F}_{m}^{(k)}$ of
$\mathcal{C}_{m}$ whereas both $\lambda^{(k)}/M$ and $r^{m}$ lie
on the face $\mathcal{F}_{m}^{(k)}$. So,  $\frac{1}{\sqrt{M}}\left\Vert \mathbf{q}_{||}^{(k)}\right\Vert \underset{j=1}{\overset{J}{\sum}}c_{j}\left(\frac{\lambda_{j}^{(k)}}{M}-r_{j}^{m}\right)=0$.
Equation (\ref{eq:delta}) is true
because $\underset{j=1}{\overset{J}{\sum}}q_{\bot j,m}^{(\epsilon,k)}\left(\frac{\lambda_{j}^{(k)}}{M}-r_{j}^{m}\right)$
is inner product in $\mathbb{R}_{+}^{J}$ which is minimized when
$r^{m}$ is chosen to be on the boundary of $\mathcal{B}_{\delta^{(k)}}^{(k)}$
so that $\left(\frac{\lambda_{j}^{(k)}}{M}-r_{j}^{m}\right)_{j}$
points in the opposite direction to $\left(q_{\bot j,m}^{(\epsilon,k)}\right)_{j}$.
Since $\\\left(\overset{M}{\underset{m=1}{\sum}}\sqrt{\underset{j=1}{\overset{J}{\sum}}\left(q_{\bot j,m}^{(\epsilon,k)}\right)^{2}}\right)^{2}\geq\overset{M}{\underset{m=1}{\sum}}\underset{j=1}{\overset{J}{\sum}}\left(q_{\bot j,m}^{(\epsilon,k)}\right)^{2}$,
we get (\ref{eq:norm}).

Now substituting (\ref{eq:V_parallel}) and (\ref{eq:drift_tight})
in (\ref{eq:W_perp}), we get \begin{align*}
\lefteqn{ \mathbb{E}\left[\vartriangle W_{\bot}^{(k)}(\mathbf{q}^{(\epsilon)})|\mathbf{q}^{(\epsilon)}(t)=\mathbf{q}^{(\epsilon)}\right]}\\
\leq & \frac{K_1+K_{2}}{2\left\Vert \mathbf{q}_{\bot}^{(\epsilon,k)}\right\Vert }-\delta^{(k)} \\
\leq & \frac{-\delta^{(k)}}{2} \textrm{ whenever } \left( W_{\bot}^{(k)}(\mathbf{q}^{(\epsilon)})\geq\frac{K_1+K_2}{\delta^{(k)}}\right).
\end{align*}
Moreover, since the departures in each time slot are bounded and the arrivals are finite there is a $D<\infty$ such that $\mathbb{P}\left(|\Delta Z(X)| \leq D\right)$ almost surely. Now, applying Lemma \ref{lem:Hajek}, we have the following
proposition.
\begin{prop}\label{prop:state_space}
Assuming all the servers are identical, for $\lambda^{(\epsilon)}\in\mathcal{C}$, under JSQ routing and MaxWeight
scheduling, for every $k\in\mathcal{K}_{\lambda^{(\epsilon)}}^{o}$,
there exists a set of finite constants $\{N_{r}^{(k)}\}_{r=1,2,...}$ such
that $\mathbb{E}\left[\left\Vert \mathbf{q}_{\bot}^{(\epsilon,k)}\right\Vert ^{r}\right]\leq N_{r}^{(k)}$
for all $\epsilon>0$ and for each $r=1,2,...$.
\end{prop}

As in \cite{Sto_04,erysri-heavytraffic}, note that $k\in\mathcal{K}_{\lambda^{(\epsilon)}}^{o}$ is an important assumption here. If $k\in\mathcal{K}\smallsetminus\mathcal{K}_{\lambda^{(\epsilon)}}^{o}$, i.e., if the arrival rate approaches a corner point of the capacity region as $\epsilon^{(k)}\rightarrow0$, then there is no constant $\delta^{(k)}$ so that $\mathcal{B}_{\delta^{(k)}}^{(k)}$ lies in the face $\mathcal{F}^{(k)}$. In other words, the $\delta^{(k)}$  depends on $\epsilon^{(k)}$ and so the bound obtained by Lemma \ref{lem:Hajek} also depends on $\epsilon^{(k)}$.

\emph{Remark:} As stated in Proposition \ref{prop:main}, our results hold only for the case of identical servers, which is the most practical scenario. However, we have written the proofs more generally whenever we can so that it is clear where we need the identical server assumption. In particular, in this subsection, up to Equation (\ref{eq:maxwt}), we do not need this assumption, but we have used the assumption after that, in analyzing the drift of $V(\mathbf{q})$. The upper bound in the next section is valid more generally if one can establish state-space collapse for the non-identical server case. However, at this time, this is an open problem.

\subsection{Upper Bound}\label{sub:jsqmw_ub}
In this section, we will obtain an upper bound on the weighted queue length, $\mathbb{E}\left[\left\langle \mathbf{c}^{(k)},\mathbf{q}^{(\epsilon)}\right\rangle \right]$ in steady state, and show that in the asymptotic limit as $\epsilon^{(k)}\downarrow0$, this coincides with the lower bound.

%

Noting that the drift of $\Delta W_{||}^{(k)}$ is zero in steady state, it can be shown, as in Lemma $8$ from \cite{erysri-heavytraffic} that in steady state, for  any $\mathbf{c}\in\mathbb{R}_{+}^{JM}$,
we have\begin{align}
\lefteqn{\mathbb{E}\left[\left\langle \mathbf{c},\mathbf{q}(t)\right\rangle \left\langle \mathbf{c},\overline{\mathbf{s}}(t)-\mathbf{a}(t)\right\rangle \right]}
\label{eq:UB_lhs}\\
= & \frac{\mathbb{E}\left[\left\langle \mathbf{c},\overline{\mathbf{s}}(t)-\mathbf{a}(t)\right\rangle ^{2}\right]}{2}+\frac{\mathbb{E}\left[\left\langle \mathbf{c},\overline{\mathbf{u}}(t)\right\rangle ^{2}\right]}{2}
\label{eq:UB_rhs1}\\
& +\mathbb{E}\left[\left\langle \mathbf{c},\mathbf{q}(t)+\mathbf{a}(t)-\overline{\mathbf{s}}(t)\right\rangle \left\langle \mathbf{c},\overline{\mathbf{u}}(t)\right\rangle \right]
\label{eq:UB_rhs2}
\end{align}
We will obtain an upper bound on $\mathbb{E}\left[\left\langle \mathbf{c}^{(k)},\mathbf{q}^{(\epsilon)}\right\rangle \right]$
by bounding each of the above terms. Before that, we need the following
definitions and results.

Let $\pi^{(k)}$ be the steady-state probability that the MaxWeight schedule chosen is from
the face $\mathcal{F}^{(k)}$, i.e.,\[
\pi^{(k)}=\mathbb{P}\left(\left\langle c,\overline{s}(t)\right\rangle =b^{(k)}\right).\]
where $\overline{s}_{j}=\underset{m=1}{\overset{M}{\sum}}\overline{s_{j}^{m}}$ as defined in (\ref{eq:s_j}). Also, define\[
\gamma^{(k)}=\min\left\{ b^{(k)}-\left\langle c,r\right\rangle :r\in\mathcal{S}\setminus\mathcal{F}^{(k)}\right\} .\]
Then noting that in steady state, \[
\mathbb{E}\left[\left\langle c^{(k)},\overline{s(q)}\right\rangle\right] \geq \left\langle c^{(k)},\lambda^\epsilon\right\rangle = b^{(k)}-\epsilon^{(k)}
,\] it can be shown as in Claim $1$ in \cite{erysri-heavytraffic} that for
for any $\epsilon^{(k)}\in\left(0,\gamma^{(k)}\right),$\begin{equation}
\left(1-\pi^{(k)}\right)\leq\frac{\epsilon^{(k)}}{\gamma^{(k)}}.\nonumber \end{equation}
Then, note that \begin{align}
\lefteqn{\mathbb{E}\left[\left(b^{(k)}-\left\langle c,\overline{s}(t)\right\rangle \right)^{2}\right]} \nonumber \\
 =& \left(1-\pi^{(k)}\right)\mathbb{E}\left[\left(b^{(k)}-\left\langle c,\overline{s}(t)\right\rangle \right)^{2}|\left(\left\langle c,\overline{s}(t)\right\rangle \neq b^{(k)}\right),\right]\nonumber \\
  \leq & \frac{\epsilon^{(k)}}{\gamma^{(k)}}\left(\left(b^{(k)}\right)^{2}+\left\langle c,s_{max}1\right\rangle ^{2}\right)\label{eq:b_cs2}\end{align}

Define $\widetilde{\mathcal{C}}_{m}\subseteq\mathbb{R}_{+}^{JM}$
as $\widetilde{\mathcal{C}}_{m}=\mathcal{C}_{1}\times...\times\mathcal{C}_{M}$.
Then, $\widetilde{\mathcal{C}}_{m}$ is a convex polygon.
\begin{claim}
\label{cla:C_m_tilde}Let $q^{m}\in\mathbb{R}_{+}^{J}$ for each
$m\in\{1,2,....M\}$. Denote $\mathbf{q}=\left(q^{m}\right)_{m=1}^{M}\in\mathbb{R}_{+}^{JM}$.
If, for each $m$, $\left(s^{m}\right)^{*}$ is a solution of $\underset{s\in\mathcal{C}_{m}}{\max}\left\langle q^{m},s\right\rangle $
 then $\mathbf{s}^{*}=(\left(s^{m}\right)^{*})_{m}$
is a solution of $\underset{\mathbf{s}\in\widetilde{\mathcal{C}}_{m}}{\max}\left\langle \mathbf{q},\mathbf{s}\right\rangle $.\end{claim}
\begin{proof}
\sloppy
Since $\mathbf{s}^{*}\in\widetilde{\mathcal{C}}_{m}$, $\left\langle \mathbf{q},\mathbf{s}^{*}\right\rangle \leq\underset{\mathbf{s}\in\widetilde{\mathcal{C}}_{m}}{\max}\left\langle \mathbf{q},\mathbf{s}\right\rangle $.
Note that $\underset{\mathbf{s}\in\widetilde{\mathcal{C}}_{m}}{\max}\left\langle \mathbf{q},\mathbf{s}\right\rangle =\overset{M}{\underset{m=1}{\sum}}\underset{s^{m}\in\mathcal{C}_{m}}{\max}\left\langle q^{m},s^{m}\right\rangle $
. Therefore, if $\left\langle \mathbf{q},\mathbf{s}^{*}\right\rangle <\underset{\mathbf{s}\in\widetilde{\mathcal{C}}_{m}}{\max}\left\langle \mathbf{q},\mathbf{s}\right\rangle $,
we have $\overset{M}{\underset{m=1}{\sum}}\left\langle q^{m},\left(s^{m}\right)^{*}\right\rangle <\overset{M}{\underset{m=1}{\sum}}\underset{s^{m}\in\mathcal{C}_{m}}{\max}\left\langle q^{m},s^{m}\right\rangle $.
Then there exists an $m\leq M$ such that $\left\langle q^{m},\left(s^{m}\right)^{*}\right\rangle <  \underset{s^{m}\in\mathcal{C}_{m}}{\max}\left\langle q^{m},s^{m}\right\rangle $,
which is a contradiction.
\end{proof}
Therefore, choosing a MaxWeight schedule at each server is same as
choosing a MaxWeight schedule from the convex polygon, $\widetilde{\mathcal{C}}_{m}$.
Since there are a finite number of feasible schedules, given $\mathbf{c}^{(k)}\in\mathbb{R}_{+}^{JM}$
such that $||\mathbf{c}^{(k)}||=1$, there exists an angle $\theta^{(k)}\in(0,\frac{\pi}{2}]$
such that,  for all $\mathbf{q}\in\left\{ \mathbf{q}\in\mathbb{R}_{+}^{JM}:||\mathbf{q}_{||}^{(k)}||\geq||\mathbf{q}||\cos\left(\theta^{(k)}\right)\right\}$, (i.e., for all $\mathbf{q}\in\mathbb{R}_{+}^{JM}$ such that $\theta_{\mathbf{q}\mathbf{q}_{||}^{(k)}}\leq \theta^{(k)}$ where  $\theta_{\mathbf{a}\mathbf{b}}$ represents the angle between vectors $\mathbf{a}$ and $\mathbf{b}$), we have
 \[
\left\langle \mathbf{c}^{(k)},\overline{\mathbf{s}}(t)\right\rangle \mathcal{I}\left(\mathbf{q}(t)=\mathbf{q}\right)=b^{(k)}/\sqrt{M}.\]

We can bound the unused service as follows.
\begin{align}
\mathbb{E}\left[\left\langle \mathbf{c}^{(k)},\overline{\mathbf{u}}(t)\right\rangle \right]
 & \leq\mathbb{E}\left[\left\langle \mathbf{c}^{(k)},\overline{\mathbf{s}}(t)-\mathbf{a}(t)\right\rangle \right]\nonumber\\
 & =\frac{1}{\sqrt{M}}\left(\mathbb{E}\left[\left\langle c^{(k)},\overline{s}(t)\right\rangle \right]-\left\langle c^{(k)},\lambda^\epsilon\right\rangle \right)\nonumber\\
 & =\frac{1}{\sqrt{M}}\left(\mathbb{E}\left[\left\langle c^{(k)},\overline{s}(t)\right\rangle \right]-\left(b^{(k)}-\epsilon^{(k)}\right)\right)\nonumber\\
 & \leq\frac{\epsilon^{(k)}}{\sqrt{M}}\label{eq:c_u}\end{align}
\sloppy where the last inequality follows from the fact that the MaxWeight schedule lies inside the capacity region and so $\mathbb{E}\left[\left\langle c^{(k)},\overline{s}(t)\right\rangle \right]\leq b^{(k)}.$

\fussy Now, we will bound each of the terms in (\ref{eq:UB_rhs2}). Let us first consider the term in (\ref{eq:UB_lhs}). Given that the arrival rate if $\lambda^\epsilon$ we have,
\begin{align*}
\lefteqn{\mathbb{E}\left[\left\langle \mathbf{c}^{(k)},\mathbf{q}(t)\right\rangle \left\langle \mathbf{c}^{(k)},\overline{\mathbf{s}}(t)-\mathbf{a}(t)\right\rangle \right]}\\
= & \mathbb{E}\left[\left\langle \mathbf{c}^{(k)},\mathbf{q}(t)\right\rangle \right]\left(\frac{b^{(k)}}{\sqrt{M}}-\frac{1}{\sqrt{M}}\left\langle c^{(k)},\lambda\right\rangle \right)\\
&-\mathbb{E}\left[\left\langle \mathbf{c}^{(k)},\mathbf{q}(t)\right\rangle \left(\frac{b^{(k)}}{\sqrt{M}}-\left\langle \mathbf{c}^{(k)},\overline{\mathbf{s}}(t)\right\rangle \right)\right]\\
= & \frac{\epsilon^{(k)}}{\sqrt{M}}\mathbb{E}\left[\left\langle \mathbf{c}^{(k)},\mathbf{q}(t)\right\rangle \right]\\
&-\mathbb{E}\left[||\mathbf{q}_{||}^{(k)}(t)||\left(\frac{b^{(k)}}{\sqrt{M}}-\left\langle \mathbf{c}^{(k)},\overline{\mathbf{s}}(t)\right\rangle \right)\right].\end{align*}
Now, we will bound the last term in this equation using the definition of $\theta^{(k)}$ as follows.\begin{align}
\lefteqn{\mathbb{E}\left[||\mathbf{q}_{||}^{(k)}(t)||\left(\frac{b^{(k)}}{\sqrt{M}}-\left\langle \mathbf{c}^{(k)},\overline{\mathbf{s}}(t)\right\rangle \right)\right]}\nonumber\\
= & \mathbb{E}\left[||\mathbf{q}(t)||\cos\left(\theta_{\mathbf{q}\mathbf{q}_{||}^{(k)}}\right)\left(\frac{b^{(k)}}{\sqrt{M}}-\left\langle \mathbf{c}^{(k)},\overline{\mathbf{s}}(t)\right\rangle \right)\right]\nonumber\\
= & \mathbb{E}\left[||\mathbf{q}(t)||\cos\left(\theta_{\mathbf{q}\mathbf{q}_{||}^{(k)}}\right)\mathcal{I}\left(\theta_{\mathbf{q}\mathbf{q}_{||}^{(k)}}>\theta^{(k)}\right)\right.\nonumber\\
&\left.\times\left(\frac{b^{(k)}}{\sqrt{M}}-\left\langle \mathbf{c}^{(k)},\overline{\mathbf{s}}(t)\right\rangle \right)\right]
\label{eq:theta}\\
= & \mathbb{E}\left[||\mathbf{q}_{\bot}^{(k)}(t)||\cot\left(\theta_{\mathbf{q}\mathbf{q}_{||}^{(k)}}\right)\mathcal{I}\left(\theta_{\mathbf{q}\mathbf{q}_{||}^{(k)}}>\theta^{(k)}\right)\right.\nonumber\\
&\left.\times\left(\frac{b^{(k)}}{\sqrt{M}}-\left\langle \mathbf{c}^{(k)},\overline{\mathbf{s}}(t)\right\rangle \right)\right]\nonumber\\
= & \mathbb{E}\left[||\mathbf{q}_{\bot}^{(k)}(t)||\mathcal{I}\left(\theta_{\mathbf{q}\mathbf{q}_{||}^{(k)}}>\theta^{(k)}\right)\left(\frac{b^{(k)}}{\sqrt{M}}-\left\langle \mathbf{c}^{(k)},\overline{\mathbf{s}}(t)\right\rangle \right)\right]\nonumber\\
&\times\cot\left(\theta^{(k)}\right)\nonumber\\
\leq & \frac{1}{\sqrt{M}}\mathbb{E}\left[||\mathbf{q}_{\bot}^{(k)}(t)||\left(b^{(k)}-\left\langle c^{(k)},s(t)\right\rangle \right)\right]\cot\left(\theta^{(k)}\right)\label{eq:cs_def}\\
\leq & \frac{\cot\left(\theta^{(k)}\right)}{\sqrt{M}}\sqrt{\mathbb{E}\left[||\mathbf{q}_{\bot}^{(k)}(t)||^{2}\right]\mathbb{E}\left[\left(b^{(k)}-\left\langle c^{(k)},s(t)\right\rangle \right)^{2}\right]}\label{eq:cauchy_scwartz}\\
\leq & \frac{\cot\left(\theta^{(k)}\right)}{\sqrt{M}}\sqrt{N_{2}^{(k)}\frac{\epsilon^{(k)}}{\gamma^{(k)}}\left(\left(b^{(k)}\right)^{2}+\left\langle c,s_{max}1\right\rangle ^{2}\right)}\nonumber\end{align}
where (\ref{eq:theta}) follows from the definition of $\theta^{(k)}$, (\ref{eq:cs_def}) follows from our choice of $\mathbf{c}^{(k)}$ and definition of $s$, (\ref{eq:cauchy_scwartz}) follows from Cauchy-Schwarz inequality. the last inequality follows from state-space collapse (Proposition \ref{prop:state_space}) and (\ref{eq:b_cs2}). Thus, we have \begin{align}
\lefteqn{\mathbb{E}\left[\left\langle \mathbf{c}^{(k)},\mathbf{q}(t)\right\rangle \left\langle \mathbf{c}^{(k)},\overline{\mathbf{s}}(t)-\mathbf{a}(t)\right\rangle \right]} \nonumber \\
\geq & \frac{\epsilon^{(k)}}{\sqrt{M}}\mathbb{E}\left[\left\langle \mathbf{c}^{(k)},\mathbf{q}(t)\right\rangle \right]\nonumber \\
 & -\frac{\cot\left(\theta^{(k)}\right)}{\sqrt{M}}\sqrt{N_{2}^{(k)}\frac{\epsilon^{(k)}}{\gamma^{(k)}}\left(\left(b^{(k)}\right)^{2}+\left\langle c,s_{max}1\right\rangle ^{2}\right)}\label{eq:LHS}\end{align}

Now, consider the first term in (\ref{eq:UB_rhs1}). Again, using the fact that the arrival rate is $\lambda^\epsilon$ we have,
\begin{align}
\lefteqn{\mathbb{E}\left[\left\langle \mathbf{c}^{(k)},\overline{\mathbf{s}}(t)-\mathbf{a}(t)\right\rangle ^{2}\right]}\nonumber \\
= & \mathbb{E}\left[\hspace{-2pt} \left( \hspace{-2pt} \left\langle \mathbf{c}^{(k)},\mathbf{a}(t)\right\rangle \hspace{-2pt} - \hspace{-2pt} \frac{b^{(k)}}{\sqrt{M}}\right)^{2}\right]+\mathbb{E}\left[\hspace{-2pt} \left(\frac{b^{(k)}}{\sqrt{M}}-\left\langle \mathbf{c}^{(k)},\overline{\mathbf{s}}(t)\right\rangle \right)^{2}\right]\nonumber \\
 & -2\frac{\epsilon^{(k)}}{\sqrt{M}}\mathbb{E}\left[\left(\frac{b^{(k)}}{\sqrt{M}}-\left\langle \mathbf{c}^{(k)},\overline{\mathbf{s}}(t)\right\rangle \right)\right]\nonumber \\
\leq & \mathbb{E}\left[\left(\frac{1}{\sqrt{M}}\left\langle c^{(k)},a(t)-\lambda^\epsilon\right\rangle +\frac{\left\langle c^{(k)},\lambda^\epsilon\right\rangle -b^{(k)}}{\sqrt{M}}\right)^{2}\right]\nonumber\\
&+\mathbb{E}\left[\left(\frac{b^{(k)}}{\sqrt{M}}-\left\langle \mathbf{c}^{(k)},\overline{\mathbf{s}}(t)\right\rangle \right)^{2}\right]\nonumber \\
= & \frac{1}{M}\mathbb{E}\left[\left(\left\langle c^{(k)},a(t)-\lambda^\epsilon\right\rangle \right)^{2}\right]+2\frac{\epsilon^{(k)}}{\sqrt{M}}\mathbb{E}\left[\left\langle c^{(k)},a(t)-\lambda^\epsilon\right\rangle \right]\nonumber \\
 &+\frac{\left(\epsilon^{(k)}\right)^{2}}{M} +\frac{1}{M}\mathbb{E}\left[\left(b^{(k)}-\left\langle c^{(k)},s(t)\right\rangle \right)^{2}\right] \nonumber \\
\leq & \frac{1}{M}\left\langle \left(c^{(k)}\right)^{2},\sigma^{2}\right\rangle +\frac{\left(\epsilon^{(k)}\right)^{2}}{M}\nonumber\\
&+\frac{1}{M}\frac{\epsilon^{(k)}}{\gamma^{(k)}}\left(\left(b^{(k)}\right)^{2}+\left\langle c,s_{max}1\right\rangle ^{2}\right)\label{eq:variance} \\
= & \frac{1}{\sqrt{M}}\left(\zeta^{(\epsilon,k)}+\frac{1}{\sqrt{M}}\frac{\epsilon^{(k)}}{\gamma^{(k)}}\left(\left(b^{(k)}\right)^{2}+\left\langle c,s_{max}1\right\rangle ^{2}\right)\right)\label{eq:first_term}\end{align}
\sloppy where $\zeta^{(\epsilon,k)}$ was defined as $\zeta^{(\epsilon,k)}=\frac{\left(\epsilon^{(k)}\right)^{2}}{\sqrt{M}}+\frac{1}{\sqrt{M}}\left\langle \left(c^{(k)}\right)^{2},\left(\sigma^{(\epsilon)}\right)^{2}\right\rangle $. Equation (\ref{eq:variance}) is obtained by noting that $\mathbb{E}\left[a(t) \right]=\lambda^\epsilon$ and so $\mathbb{E}\left[\left(\left\langle c^{(k)},a(t)-\lambda^\epsilon\right\rangle \right)^{2}\right] = var\left(\left\langle c^{(k)},a(t)-\lambda^\epsilon\right\rangle \right)=\left\langle c^{(k)},var(a(t)-\lambda^\epsilon)\right\rangle$.

\fussy Consider the second term in (\ref{eq:UB_rhs1}).
\begin{align}
\mathbb{E}\left[\left\langle \mathbf{c}^{(k)},\overline{\mathbf{u}}(t)\right\rangle ^{2}\right]
 & \leq\left\langle \mathbf{c}^{(k)},\mathbf{1}s_{max}\right\rangle \mathbb{E}\left[\left\langle \mathbf{c}^{(k)},\overline{\mathbf{u}}(t)\right\rangle \right]\nonumber \\
 & \leq\frac{\epsilon^{(k)}}{\sqrt{M}}\left\langle \mathbf{c}^{(k)},\mathbf{1}s_{max}\right\rangle \label{eq:second_term}\end{align}
where the last inequality follows from (\ref{eq:c_u}).

Now, we consider the term in (\ref{eq:UB_rhs2}). We need some definitions
so that we can only consider the non-zero components of $c$. Let
$\mathcal{L}_{++}^{(k)}=\left\{ j\in\left\{ 1,2,...J\right\} :c_{j}^{(k)}>0\right\} $.
Define $\widetilde{\mathbf{c}}^{(k)}=\left(c_{jm}^{(k)}\right)_{j\in\mathcal{L}_{++}^{(k)}}$
,$\widetilde{\mathbf{q}}=\left(q_{jm}\right)_{j\in\mathcal{L}_{++}^{(k)}}$
and $\widetilde{\mathbf{u}}=\left(\overline{u}_{jm}\right)_{j\in\mathcal{L}_{++}^{(k)}}$.
Also define, the projections, $\widetilde{\mathbf{q}}^{(k)}_{||}=\left\langle \widetilde{\mathbf{c}}^{(k)},\widetilde{\mathbf{q}}\right\rangle \widetilde{\mathbf{c}}^{(k)}$
and $\widetilde{\mathbf{q}}^{(k)}_{\bot}=\widetilde{\mathbf{q}}-\widetilde{\mathbf{q}}^{(k)}_{||}$.
Similarly, define $\widetilde{\mathbf{u}}^{(k)}_{||}$ and $\widetilde{\mathbf{u}}^{(k)}_{\bot}$.
Then, we have
 \begin{align}
\lefteqn{\mathbb{E}\left[\left\langle \mathbf{c}^{(k)},\mathbf{q}(t)+\mathbf{a}(t)-\overline{\mathbf{s}}(t)\right\rangle \left\langle \mathbf{c}^{(k)},\overline{\mathbf{u}}(t)\right\rangle \right]}\nonumber \\
= & \mathbb{E}\left[\left\langle \mathbf{c}^{(k)},\mathbf{q}(t+1)\right\rangle \left\langle \mathbf{c}^{(k)},\overline{\mathbf{u}}(t)\right\rangle \right]-\mathbb{E}\left[\left\langle \mathbf{c}^{(k)},\overline{\mathbf{u}}(t)\right\rangle ^{2}\right]\nonumber \\
\leq & \mathbb{E}\left[\left\langle \mathbf{c}^{(k)},\mathbf{q}(t+1)\right\rangle \left\langle \mathbf{c}^{(k)},\overline{\mathbf{u}}(t)\right\rangle \right]\nonumber \\
= & \mathbb{E}\left[\left\langle \widetilde{\mathbf{c}}^{(k)},\widetilde{\mathbf{q}}(t+1)\right\rangle \left\langle \widetilde{\mathbf{c}}^{(k)},\widetilde{\mathbf{u}}(t)\right\rangle \right]\nonumber \\
= & \mathbb{E}\left[ ||\widetilde{\mathbf{q}}^{(k)}_{||}(t+1)||||\widetilde{\mathbf{u}}^{(k)}_{||}|| \right]\nonumber \\
= & \mathbb{E}\left[ \left\langle \widetilde{\mathbf{q}}^{(k)}_{||}(t+1),\widetilde{\mathbf{u}^{(k)}_{||}}(t)\right\rangle \right]\nonumber \\
= & \mathbb{E}\left[ \left\langle \widetilde{\mathbf{q}}^{(k)}_{||}(t+1),\widetilde{\mathbf{u}}(t)\right\rangle \right]\nonumber \\
= & \mathbb{E}\left[ \left\langle \widetilde{\mathbf{q}}(t+1),\widetilde{\mathbf{u}}(t)\right\rangle \right] +
\mathbb{E}\left[ \left\langle -\widetilde{\mathbf{q}}^{(k)}_{\bot}(t+1),\widetilde{\mathbf{u}}(t)\right\rangle \right]
\nonumber \\
\leq & \mathbb{E}\left[ \left\langle D_{max}s_{max}\mathbf{1},\widetilde{\mathbf{u}}(t)\right\rangle \right] + \sqrt{\mathbb{E}\left[||\widetilde{\mathbf{q}}^{(k)}_{\bot}(t+1)||^{2}\right]\mathbb{E}\left[||\widetilde{\mathbf{u}}(t)||^{2}\right]}
\label{eq:last_term_1}\\
\leq & D_{max}s_{max}\mathbb{E}\left[ \left\langle \mathbf{1},\widetilde{\mathbf{u}}(t)\right\rangle \right] +
\sqrt{N_{2}^{(k)}\mathbb{E}\left[\left\langle\widetilde{\mathbf{u}}(t),\widetilde{\mathbf{u}}(t)\right\rangle\right]}
\label{eq:last_term_2}\\
\leq & D_{max}s_{max}\mathbb{E}\left[ \left\langle \mathbf{1},\widetilde{\mathbf{u}}(t)\right\rangle \right] +
\sqrt{N_{2}^{(k)}s_{max}\mathbb{E}\left[\left\langle\mathbf{1},\widetilde{\mathbf{u}}(t)\right\rangle\right]}
\nonumber
\end{align}
where (\ref{eq:last_term_1}) follows from (\ref{eq:unused}) and from Cauchy-Schwarz inequality. Equation (\ref{eq:last_term_2}) follows from from state-space collapse (Proposition \ref{prop:state_space}), since  $\mathbb{E}\left[||\widetilde{\mathbf{q}}^{(k)}_{\bot}||^{2}\right]\leq\mathbb{E}\left[||\overline{\mathbf{q}}^{(k)}_{\bot}||^{2}\right]\leq N_{2}^{(k)}$.

Note that
\begin{align}
\mathbb{E}\left[\left\langle \mathbf{1},\widetilde{\mathbf{u}}(t)\right\rangle \right]
 & \leq\frac{1}{c_{min}^{(k)}}\mathbb{E}\left[\left\langle \widetilde{\mathbf{c}}^{(k)},\widetilde{\mathbf{u}}(t)\right\rangle \right]\nonumber\\
 & =\frac{1}{c_{min}^{(k)}}\mathbb{E}\left[\left\langle \mathbf{c}^{(k)},\overline{\mathbf{u}}(t)\right\rangle \right]\nonumber\\
 & \leq\frac{\epsilon^{(k)}}{\sqrt{M}}\nonumber
\end{align}
where $c_{min}^{(k)}\overset{\Delta}{=}\underset{j\in\mathcal{L}_{++}^{(k)}}{\min}c_{j}^{(k)}>0$ and the last inequality follows from (\ref{eq:c_u}). Thus, we have \begin{align}
\lefteqn{\mathbb{E}\left[\left\langle \mathbf{c}^{(k)},\mathbf{q}(t)+\overline{\mathbf{s}}(t)-\mathbf{a}(t)\right\rangle \left\langle \mathbf{c}^{(k)},\overline{\mathbf{u}}(t)\right\rangle \right]}\nonumber\\
\leq&D_{max}s_{max}\frac{\epsilon^{(k)}}{\sqrt{M}} +
\sqrt{N_{2}^{(k)}s_{max}\frac{\epsilon^{(k)}}{\sqrt{M}}}
\label{eq:last_term}
\end{align}
Now, substituting (\ref{eq:LHS}), (\ref{eq:first_term}), (\ref{eq:second_term})
and (\ref{eq:last_term}) in (\ref{eq:UB_rhs2}), we get\begin{align*}
\epsilon^{(k)}\mathbb{E}\left[\left\langle \mathbf{c}^{(k)},\mathbf{q}(t)\right\rangle \right] & \leq  \frac{\zeta^{(\epsilon,k)}}{2}+B_{2}^{(\epsilon,k)}\end{align*}
where \begin{align*}
B_{2}^{(\epsilon,k)}= & \frac{1}{2\sqrt{M}}\frac{\epsilon^{(k)}}{\gamma^{(k)}}\left(\hspace{-2pt}\left(b^{(k)}\right)^{2}+\left\langle c,s_{max}1\right\rangle ^{2}\hspace{-2pt}\right) +D_{max}s_{max}\epsilon^{(k)} \\
 & +\frac{\epsilon^{(k)}}{2}\left\langle \mathbf{c}^{(k)},\mathbf{1}s_{max}\right\rangle +
\sqrt{\sqrt{M}N_{2}^{(k)}s_{max}\epsilon^{(k)}}\\
 & +\cot\left(\theta^{(k)}\right)\sqrt{N_{2}^{(k)}\frac{\epsilon^{(k)}}{\gamma^{(k)}}\left(\left(b^{(k)}\right)^{2}+\left\langle c,s_{max}1\right\rangle ^{2}\right)}.\end{align*}
Thus, in the heavy traffic limit as $\epsilon^{(k)}\downarrow0$,
we have that \begin{equation}
\underset{\epsilon^{(k)}\downarrow0}{\lim}\epsilon^{(k)}\mathbb{E}\left[\left\langle \mathbf{c}^{(k)},\overline{\mathbf{q}}^{(\epsilon)}\right\rangle \right]\leq\frac{\zeta^{(k)}}{2}\label{eq:upper_bound}\end{equation}
where $\zeta^{(k)}$ was defined as $\zeta^{(k)}=\frac{1}{\sqrt{M}}\left\langle \left(c^{(k)}\right)^{2},\left(\sigma\right)^{2}\right\rangle $.
Thus, (\ref{eq:lower_bound}) and (\ref{eq:upper_bound}) establish
the first moment heavy-traffic optimality of JSQ routing and MaxWeight
scheduling policy.  The proof of Proposition \ref{prop:main} is now complete.

\subsection{Power-of-Two-Choices Routing and MaxWeight Scheduling}

JSQ routing needs complete
queue length information at the router. In practice, this communication
overhead can be considerable when the number of servers is large. An
alternate algorithm is the power-of-two-choices routing algorithm.

In this algorithm, in each time
slot $t$, for each type of job $m$, two servers $m_{1}^{j}(t)$
and $m_{2}^{j}(t)$ are chosen uniformly at random. All the type $m$
job arrivals in this time slot are then routed to the server with
the shorter queue length among these two, i.e., $m_{j}^{*}(t)=\underset{m\in\{m_{1}^{j}(t),m_{2}^{j}(t)\}}{\arg\min}q_{j,m}(t)$.

It was shown in \cite{magsriyin12} that power-of-two-choices routing algorithm with MaxWeight scheduling  is throughput optimal if all the servers are identical. From the proof of throughput optimality, one obtains
\begin{align}
\lefteqn{\mathbb{E}\left[\vartriangle V(\mathbf{q}^{(\epsilon)})|\mathbf{q}^{(\epsilon)}(t)=\mathbf{q}^{(\epsilon)}\right]}\nonumber\\
\leq & K'+\underset{j=1}{\overset{J}{\sum}}2\lambda_{j}\overset{M}{\underset{m=1}{\sum}}\frac{q_{j,m}^{(\epsilon)}}{M}-\mathbb{E}\left[\overset{M}{\underset{m=1}{\sum}}\underset{j=1}{\overset{J}{\sum}}2q_{j,m}^{(\epsilon)}s_{j}^{m}(t)\right] \nonumber \end{align}
Note that this inequality is identical to (\ref{eq:ref_po2_maxwt}), in the proof of state-space collapse of JSQ routing and MaxWeight scheduling policy. Also note that the remainder of the proof of state-space collapse and upper bound in Sections \ref{sub:jsqmw_state_space} and \ref{sub:jsqmw_ub} is independent of the routing policy. Moreover, the proof of lower bound in Section \ref{sub:lower_bound} is also valid here. Thus, once we have the above relation, the proof of heavy traffic optimality of this policy is identical to that of JSQ routing and MaxWeight scheduling policy.

\section{Power-of-Two-Choices routing}\label{sec:powerof2}

In this section, we consider the power-of-two-choices routing algorithm, without any scheduling. This is a special case of the model considered in the previous section when all the jobs are of the same type. In this case, there is a single queue at each server and no scheduling is needed.
\subsection*{Note on Notation}
 In this section, since $J=1$ here, we just denote all vectors ( in $\mathbb{R}^{M}$) in bold font $\mathbf{x}$.

The result from previous section is not applicable here because of the following reason. In Proposition \ref{prop:main}, a sequence of systems with arrival rate approaching a face of the capacity region, along its normal vector were considered. The normal vector of the face plays an important role in the state space collapse, and so the upper bound obtained is in terms of this normal. So, this result cannot be applied if the arrival rates were approaching a corner point where there is no common normal vector.
In particular, the proof of state space collapse in Section \ref{sub:jsqmw_state_space} is not applicable here because one cannot define a ball $\mathcal{B}_{\delta^{(k)}}^{(k)}$ as in  (\ref{eq:ball}) at a corner point.

Let ${\cal A}(t)$ denote the set of jobs that arrive at the beginning of time slot $t.$ Let $D_k$ be the size of $k^{th}$ job. We define $a(t)=\sum_{k \in{\cal A}(t)} D_k,$ to be the overall size of the jobs in ${\cal A}(t)$ or the total time slots requested by the jobs. We assume that $a(t)$ is a stochastic process which is i.i.d. across time slots, $\mathbb{E}[a(t)]=\lambda$ and $\Pr(a(t)=0)>\epsilon_a$ for some $\epsilon_a>0$ for all $t$. Let $\sigma^{2}=var[a(t)]$.
Let $X(t)$ denote the servers chosen
at time slot $t$. So, $X(t)$ can take one of $^{M}C_{2}$ values
of the form $(m,m')$ where $m,m'\in\mathbb{Z}_{+}$ and $1\leq m<m'\leq M$.
Here $^{M}C_{2}$ denotes the number of $2$-combinations in a set of size $M$.
Note that $X(t)$ is an i.i.d. random process with a uniform distribution
over all possible values. Define $^{M}C_{2}$ different arrival processes
denoted by $a_{m,m'}(t)$ with $1\leq m<m'\leq M$ as follows. If $x(t)=(\hat{m},\hat{m'})$,
then \[
a_{m,m'}(t)=\begin{cases}
a(t) & \textrm{ for }m=\hat{m}\textrm{ and }m'=\hat{m'}\\
0 & \textrm{otherwise}\end{cases}.\]
Thus, $\{a_{m,m'}(t)\}$ can be thought of as a set of correlated arrival
processes. They are correlated so that only one of them can have a
non-zero value at each time. Let $\lambda_{m,m'}=\mathbb{E}[a_{m,m'}(t)]$.
Then $\lambda_{m,m'}=\frac{\lambda}{^{M}C_{2}}$. The arrivals in $a_{m,m'}(t)$
can be routed only to either server $m$ or server $m'$. According
to the power-of-two-choices algorithm, all the jobs are then routed
to the server with smallest queue among $m$ and $m'$. Ties are broken
at random. Let $a_{m}(t)$ denote the arrivals to server $m$ at time
$t$ after routing.

Let $\mu$ be the amount of service available in each time slot at each server. Not all of this service may be used either because the queue is empty or because different chunks of same job cannot be served simultaneously. Let $s_{m}(t)$ be the actual amount of service scheduled available in time slot $t$ at server $m$. Let $u_{m}(t)$ denote the unused service which is defined as
$u_{m}(t)=\mu-s_{m}(t)$. Let $q_{m}(t)$ denote the queue length
at server $m$ at time $t$, and let $\mathbf{q}(t)$ denote the vector
$(q_{1}(t),q_{2}(t),....q_{M}(t))$ Then, we have \[
q_{m}(t+1)=q_{m}(t)+a_{m}(t)-\mu+u_{m}(t).\]
Note that
\begin{equation}
    u_{m}(t)=0 \textrm{ whenever } q_{m}(t)+a_m(t) \geq D_{max}\mu.\label{eq:unused_po2}
\end{equation}


We again follow the procedure used in the previous section to show heavy traffic optimality. Since power-of-two-choices algorithm tries to equalize any two randomly chosen queues, we expect that there is a state-space collapse along the direction where all queues are equal, similar to JSQ algorithm.

Let $\mathbf{c}_{1}=\frac{1}{\sqrt{M}}(1,1,.....1)$
be the unit vector in $\mathbb{R}^{M}$ along which we expect state-space collapse. Let $\mathbf{1}$ denote
the vector (1,1,.....1). For any $\mathbf{Q}\in\mathbb{R}^{M}$,
define $\mathbf{Q}_{||}$ to be the component of $\mathbf{Q}$ along
$\mathbf{c}_{1}$, i.e., $\mathbf{Q}_{||}=\left\langle \mathbf{Q},\mathbf{c}_{1}\right\rangle \mathbf{c}_{1}$
where $\left\langle .,.\right\rangle $ denotes the canonical dot
product. Thus, $\mathbf{Q}_{||}=\frac{\underset{m}{\sum}Q_{m}}{M}\mathbf{1}$. Define $\mathbf{Q}_{\bot}$ to be the component of $\mathbf{Q}$ perpendicular to $\mathbf{Q}_{||}$, i.e., $\mathbf{Q}_{\bot}=\mathbf{Q}-\mathbf{Q}_{||}$.

Define the Lyapunov functions $V_{||}(\mathbf{Q})=||\mathbf{Q}_{||}||^{2}=\frac{\left(\underset{m}{\sum}Q_{m}\right)^{2}}{M}$ and $W_{\bot}(\mathbf{Q})=||\mathbf{Q}_{\bot}||=\left[\underset{m}{\sum}Q_{m}^{2}-\frac{\left(\underset{m}{\sum}Q_{m}\right)^{2}}{M}\right]^{\frac{1}{2}}$.

\subsection{Lower Bound}\label{sub:p2_Lower_Bound}

Consider an arrival process with arrival rate $\lambda^{(\epsilon)}$
such that $\epsilon=M\mu-\lambda^{(\epsilon)}$. Let $\mathbf{q}^{(\epsilon)}(t)$
denote the corresponding queue length vector. Since the system is stabilizable,
there exists a steady-state distribution of $\mathbf{q}^{(\epsilon)}(t)$.
Again, lower bounding $(\underset{m}{\sum}\overline{\mathbf{q}}^{(\epsilon)})$ by a single queue length as in Section \ref{sub:lower_bound}, we  have \begin{equation}
\mathbb{E}\left[\underset{m}{\sum}\overline{\mathbf{q}}^{(\epsilon)}\right]\geq\frac{\left(\sigma^{(\epsilon)}\right)^{2}+\epsilon^{2}}{2\epsilon}-B_{1}\nonumber\end{equation}
where $B_{1}=\frac{Ms_{max}}{2}$ . Thus, in the heavy-traffic limit
we have \begin{equation}
\underset{\epsilon\rightarrow0}{\liminf}\epsilon\mathbb{E}\left[\underset{m}{\sum}\overline{\mathbf{q}}^{(\epsilon)}\right]\geq\frac{\sigma^{2}}{2}.\label{eq:p2_lower_limit}\end{equation}

\subsection{State Space Collapse}

For simplicity of notation, in this sub-section, we write $\mathbf{q}$
for $\mathbf{q}^{(\epsilon)}$. We will bound the drift of the Lyapunov function
$W_{\bot}(\mathbf{Q})$, and again use Lemma \ref{lem:Hajek} to  obtain state space collapse. We again use (\ref{eq:W_perp}) with $\mathbf{c}_{1}$ instead of $\mathbf{c}^{(k)}$ to get the drift of $W_{\bot}^{(k)}(\mathbf{q})$ in terms of drifts of $V(\mathbf{q})$ and $V_{||}^{(k)}(\mathbf{q})$.

Let us first consider the last term.
\begin{align}
\lefteqn{\mathbb{E}\left[\vartriangle V_{||}(\mathbf{q})|\mathbf{q}(t)=\mathbf{q}\right]}\nonumber \\
= & \mathbb{E}\left[V_{||}(\mathbf{q}(t+1))-V_{||}(\mathbf{q}(t))|\mathbf{q}(t)=\mathbf{q}\right]\nonumber \\
= & \frac{1}{M}\mathbb{E}\left[\left(\underset{m}{\sum}q_{m}(t+1)\right)^{2}-\left(\underset{m}{\sum}q_{m}(t)\right)^{2}|\mathbf{q}(t)=\mathbf{q}\right]\nonumber \\
= & \frac{1}{M}\mathbb{E}\left[\left(\underset{m}{\sum}q_{m}(t)+a_{m}(t)-\mu+\underset{m}{\sum}u_{m}(t)\right)^{2}\nonumber\right.\\
&\left.-\left(\underset{m}{\sum}q_{m}(t)\right)^{2}|\mathbf{q}(t)=\mathbf{q}\right]\nonumber \\
= & \frac{1}{M}\mathbb{E}\left[\left(\underset{m}{\sum}q_{m}(t)+a_{m}(t)-\mu\right)^{2}+\left(\underset{m}{\sum}u_{m}(t)\right)^{2}\right.\nonumber \\
 & \left.+2\left(\underset{m}{\sum}q_{m}(t)+a_{m}(t)-\mu\right)\left(\underset{m}{\sum}u_{m}(t)\right)\nonumber\right.\\
&\left.-\left(\underset{m}{\sum}q_{m}(t)\right)^{2}|\mathbf{q}(t)=\mathbf{q}\right] \nonumber\\
\geq & \frac{1}{M}\mathbb{E}\left[\left(\underset{m}{\sum}a_{m}(t)-\mu\right)^{2}+2\left(\underset{m}{\sum}q_{m}(t)\right)\left(\underset{m}{\sum}a_{m}(t)-\mu\right)\right.\nonumber\\
&\left.-2M\mu\left(\underset{m}{\sum}u_{m}(t)\right)|\mathbf{q}(t)=\mathbf{q}\right]\nonumber \\
\geq & \frac{2}{M}\left(\underset{m}{\sum}q_{m}\right)\mathbb{E}\left[\left(\underset{m}{\sum}a_{m}(t)-\mu\right)|\mathbf{q}(t)=\mathbf{q}\right]\nonumber\\
&-2\mu\mathbb{E}\left[\underset{m}{\sum}u_{m}(t)|\mathbf{q}(t)=\mathbf{q}\right]\nonumber \\
\geq & -K_{3}+2\left(\underset{m}{\sum}q_{m}\right)\left(\frac{\lambda}{M}-\mu\right)\nonumber \\
\geq & -K_{3}-2\frac{\epsilon}{M}\left(\underset{m}{\sum}q_{m}\right)\label{eq:p2_drift_V_parallel}
\end{align}

where $K_{3}=2M\mu^{2}$ is obtained by bounding $s_m(t)$ and $u_m(t)$ by $s_{max}$.

\sloppy
Now, we will bound the first term in (\ref{eq:W_perp}). Expanding ${[\vartriangle V(\mathbf{q})|\mathbf{q}(t)]}$ and using (\ref{eq:unused_po2}), it is easy to see that
\begin{align*}
\lefteqn{\mathbb{E}\left[\vartriangle V(\mathbf{q})|\mathbf{q}(t)=\mathbf{q}\right]}\\
\leq & K_{4}-2\mu\underset{m}{\sum}q_{m}(t)\\
 & +\mathbb{E}_{X}\mathbb{E}\left[\underset{m}{\sum}2q_{m}(t)a_{m}(t)|\mathbf{q}(t)=\mathbf{q},X(t)=i,j\right].
 \end{align*}
\fussy where $K_4=M(2\mu^2(D_{max}+1)+\sigma^2+\lambda^2)$.
Let $p$ be a permutation of $(1,2,...M)$ so that $q_{p(1)}\leq q_{p(2)}\leq.....\leq q_{p(M)}$.
Let $p'$ be the inverse permutation. In other words, $p'(m)$ is
the position of $m$ in the permutation $p$.
Let $q_{min}=q_{p(1)}$ and $q_{min}=q_{p(M)}$. Then, we have
\begin{align*}
\lefteqn{\mathbb{E}\left[\vartriangle V(\mathbf{q})|\mathbf{q}(t)=\mathbf{q}\right]}\\
\leq & K_4 -2\mu\underset{m}{\sum}q_{m}(t)+2q_{min}\frac{\lambda}{^{M}C_{2}}\\
 & +\underset{(i,j)\neq(p(1),p(M))}{\sum}\frac{1}{^{M}C_{2}}\mathbb{E}\left[q_{i}(t)a(t)+q_{i}(t)a(t)|X(t)=i,j\right]\\
= & K_4 -2\mu\underset{m}{\sum}q_{m}(t)-\frac{\lambda}{^{M}C_{2}}(q_{max}-q_{min})+\frac{2\lambda}{M}\underset{m}{\sum}q_{m}(t).\\
= & K_4 -2\frac{\epsilon}{M}\underset{m}{\sum}q_{m}(t)-\frac{\lambda}{^{M}C_{2}}(q_{max}-q_{min})\end{align*}
Note that \begin{align*}
||\mathbf{q}_{\perp}|| & =\sqrt{\underset{m}{\sum}\left(q_{m}-\frac{\underset{m}{\sum}q_{m}}{M}\right)^{2}}\\
 & \leq\sqrt{M\left(q_{max}-q_{min}\right)^{2}}\\
 & =\sqrt{M}\left(q_{max}-q_{min}\right).\end{align*}
Thus, we have, \begin{align*}
\mathbb{E}\left[\vartriangle V(\mathbf{q})|\mathbf{q}(t)=\mathbf{q}\right]
\leq & K_4 -2\frac{\epsilon}{M}\underset{m}{\sum}q_{m}(t)-\frac{\lambda}{^{M}C_{2}}\frac{||\mathbf{q}_{\perp}||}{\sqrt{M}}
\end{align*}

Substituting this and (\ref{eq:p2_drift_V_parallel}) in (\ref{eq:W_perp}),
we have \begin{align*}
\mathbb{E}\left[\vartriangle W_{\bot}(\mathbf{q})\right] & \leq\frac{K_{3}+K_{4}}{2||\mathbf{q}_{\bot}||}-\frac{\lambda}{^{M}C_{2}}\frac{1}{2\sqrt{M}}.\end{align*}
This means that we have negative drift for sufficiently large $W_{\bot}(\mathbf{q})$. Since the drift of $W_{\bot}(\mathbf{q})$ is finite with probability $1$, using Lemma \ref{lem:Hajek}, there exist finite constants
$\{N'_{r}\}_{r=1,2,...}$ such that $\mathbb{E}\left[||\overline{\mathbf{q}}^{(\epsilon)}||^{r}\right]\leq N'_{r}$
for each $r=1,2,...$.

\subsection{Upper Bound}

The upper bound is again obtained by bounding each of the terms in (\ref{eq:UB_rhs2}). This is identical to the case of JSQ routing (Proposition 3 in \cite{erysri-heavytraffic}). So, we will not repeat the proof here, but just state the upper bound.
\[
\mathbb{E}\left[\underset{m}{\sum}\overline{\mathbf{q}}^{(\epsilon)}\right]\geq\frac{\left(\sigma^{(\epsilon)}\right)^{2}+\epsilon^{2}}{2\epsilon}-B_{2}^{(\epsilon)}\]
where $B_{2}^{(\epsilon)}=M\sqrt{\frac{N_{2}s_{max}}{\epsilon}}+\frac{s_{max}}{2}$.
Thus, in heavy traffic limit, we have \begin{equation}
\underset{\epsilon\rightarrow0}{\liminf}\epsilon\mathbb{E}\left[\underset{m}{\sum}\overline{\mathbf{q}}^{(\epsilon)}\right]\geq\frac{\sigma^2}{2}.\nonumber\end{equation}
This coincides with the heavy-traffic lower bound in (\ref{eq:p2_lower_limit}). This establishes the first-moment heavy-traffic
optimality of power-of-two choices routing algorithm.

\section{Conclusions}

We considered a stochastic model for load balancing and
scheduling in cloud computing clusters. We studied the performance of JSQ routing and MaxWeight scheduling policy under this model. It was known that this policy is throughput optimal. We have shown that it is heavy traffic optimal when all the servers are identical. We also found that using the power-of-two-choices routing instead of JSQ routing is also heavy traffic optimal.

We then considered a simpler setting where the jobs are of the same type, so only load balancing is needed. It has been established by others using diffusion limit arguments  that the power-of-two-choices algorithm is heavy traffic optimal. We presented a  steady-state version of this result here using Lyapunov drift arguments.

\section{Acknowledgments}
Research was funded in part by ARO MURI W911NF-08-1-0233 and NSF grant CNS-0963807.

\bibliographystyle{abbrv}
\bibliography{references}
\end{document}